\documentclass[11pt,a4paper]{amsart}

\usepackage{tikz,pgfplots}
\usetikzlibrary{intersections, backgrounds}
\pgfplotsset{compat=1.16}
\usepgfplotslibrary{groupplots}
\usetikzlibrary{arrows.meta}
\usetikzlibrary{backgrounds}
\usepgfplotslibrary{patchplots}
\usepgfplotslibrary{fillbetween}

\calclayout

\usetikzlibrary{external, patterns}

\usepackage{xcolor}
\colorlet{RefColor}{green!50!black}
\colorlet{LinkColor}{red!50!black}
\usepackage{hyperref}
\hypersetup{
  colorlinks = true,
  linkbordercolor = {white},
  linkcolor = LinkColor,
  anchorcolor=LinkColor,
  citecolor= RefColor,
  filecolor=cyan,
  menucolor=red,
  runcolor=cyan,
  urlcolor = LinkColor,
}

\usepackage[capitalize,nameinlink]{cleveref}

\crefname{section}{section}{sections}
\crefname{subsection}{subsection}{subsections}
\Crefname{section}{Section}{Sections}
\Crefname{subsection}{Subsection}{Subsections}
\Crefname{figure}{Figure}{Figures}
\crefformat{equation}{\textup{#2(#1)#3}}
\crefrangeformat{equation}{\textup{#3(#1)#4--#5(#2)#6}}
\crefmultiformat{equation}{\textup{#2(#1)#3}}{ and \textup{#2(#1)#3}}
{, \textup{#2(#1)#3}}{, and \textup{#2(#1)#3}}
\crefrangemultiformat{equation}{\textup{#3(#1)#4--#5(#2)#6}}%
{ and \textup{#3(#1)#4--#5(#2)#6}}{, \textup{#3(#1)#4--#5(#2)#6}}{, and \textup{#3(#1)#4--#5(#2)#6}}

\Crefformat{equation}{#2Equation~\textup{(#1)}#3}
\Crefrangeformat{equation}{Equations~\textup{#3(#1)#4--#5(#2)#6}}
\Crefmultiformat{equation}{Equations~\textup{#2(#1)#3}}{ and \textup{#2(#1)#3}}
{, \textup{#2(#1)#3}}{, and \textup{#2(#1)#3}}
\Crefrangemultiformat{equation}{Equations~\textup{#3(#1)#4--#5(#2)#6}}%
{ and \textup{#3(#1)#4--#5(#2)#6}}{, \textup{#3(#1)#4--#5(#2)#6}}{, and \textup{#3(#1)#4--#5(#2)#6}}
\crefdefaultlabelformat{#2\textup{#1}#3}

\usepackage{amsfonts}
\usepackage{graphicx}
\usepackage{epstopdf}
\usepackage{algorithmic}
\usepackage{enumerate}
\usepackage{mathtools}
\usepackage{algorithm}

\ifpdf
  \DeclareGraphicsExtensions{.eps,.pdf,.png,.jpg}
\else
  \DeclareGraphicsExtensions{.eps}
\fi

\newtheorem{remark}{Remark}

\newtheorem{definition}{Definition}
\newtheorem{lemma}{Lemma}
\newtheorem{proposition}{Proposition}
\newtheorem{theorem}{Theorem}

\title[Parametric SOBMOR]{Structured Optimization-Based Model Order Reduction for Parametric Systems}

\author[P.\ Schwerdtner and M.\ Schaller]{Paul Schwerdtner$^{1}$ and Manuel Schaller$^{2}$}
\thanks{}
\thanks{$^{1}$FG Numerische Mathematik, Institute of Mathematics, Technische Universit\"at Berlin, Germany 
	{\tt\small schwerdt@math.tu-berlin.de}}%
\thanks{$^{2}$FG Optimization-based Control, Institute of Mathematics, Technische Universit\"at Ilmemau, Germany {\tt\small  manuel.schaller@tu-ilmenau.de}.\\
The research presented in this paper has been supported by the German Research Foundation (DFG) within the projects VO2243/2--1 ``Interpolation-Based Numerical Methods in Robust Control''
}

\usepackage{amsopn}
\usepackage{Notation}

\makeatletter
\newcommand{\pushright}[1]{\ifmeasuring@#1\else\omit\hfill$\displaystyle#1$\fi\ignorespaces}
\newcommand{\pushleft}[1]{\ifmeasuring@#1\else\omit$\displaystyle#1$\hfill\fi\ignorespaces}
\makeatother

\newtheorem*{myproof}{Proof}
\newtheorem{romdef}{ROM-Setup}

\begin{document}

\maketitle

\begin{abstract}
  We develop an optimization-based algorithm for parametric model order reduction~(PMOR) of linear time-invariant dynamical systems. Our method aims at minimizing the $\hinflinf$ approximation error in the frequency and parameter domain by an optimization of the reduced order model (ROM) matrices. State-of-the-art PMOR methods often compute several nonparametric ROMs for different parameter samples, which are then combined to a single parametric ROM\@. However, these parametric ROMs can have a low accuracy between the utilized sample points. In contrast, our optimization-based PMOR method minimizes the approximation error across the entire parameter domain. Moreover, due to our flexible approach of optimizing the system matrices directly, we can enforce favorable features such as a port-Hamiltonian structure in our ROMs across the entire parameter domain.

  Our method is an extension of the recently developed SOBMOR-algorithm to parametric systems. We extend both the ROM parameterization and the adaptive sampling procedure to the parametric case. Several numerical examples demonstrate the effectiveness and high accuracy of our method in a comparison with other PMOR methods.
\end{abstract}



\section{Introduction}
In this article, we extend the previously developed \emph{model order reduction} (MOR) method SOBMOR (\textbf{S}tructured \textbf{O}ptimization-\textbf{B}ased \textbf{M}odel \textbf{O}rder \textbf{R}eduction, see~\cite{SchwerdtnerV2020}) to provide a novel method for \emph{parametric model order reduction} (PMOR). 

Typically, MOR is applied when a mathematical model of a complex dynamical system has a state-space dimension that makes simulation or model-based control of this \emph{full order model} (FOM) computationally prohibitively expensive. Then, MOR provides a \emph{reduced order model} (ROM) that approximates the dynamic behavior of the FOM\@. When a large-scale model depends on a set of parameters that is not fixed at the time of reduction, the parameter dependency must be retained during MOR to compute a parametric ROM that approximates the given parametric FOM for all parameter configurations of interest. This is the goal of PMOR\@.


The review article~\cite{BennerGW2015} lists numerous applications of PMOR for accelerating design, control, and uncertainty quantification of dynamical systems. At the design stage, PMOR can be used to assess the behavior of a system for different parameter choices or at different operation points. In the control setting, PMOR enables the evaluation of the controller performance of a single controller for the entire considered parameter range or even the definition of an adaptive control strategy, in which the control laws are updated as the parameter changes. In uncertainty quantification, a large number of simulations is often required for different parameter samples. The runtimes of these repeated simulations can be decreased significantly by using a parametric ROM instead of the FOM\@. We refer to~\cite{BennerGW2015} for a collection of scientific and industrial applications of PMOR and to~\cite{BaurBHHMO2011} for a recent comprehensive comparison of different methods.

In this work, we consider the reduction of linear time-invariant parametric systems of the form
\begin{align}
  \label{eq:fullmodel}
  \Sigma(p):
  \begin{cases}
    \fE(p) \dot x(t, p) = \fA(p) x(t, p) + \fB(p)u(t), \\
    \phantom{\fE(p)} y(t, p) = \fC(p) x(t, p) + \fD(p)u(t),
  \end{cases}
\end{align}
where $p\in \pdom$ is a model parameter vector confined to a compact parameter domain $\pdom\subset \R^{\dimp}$ and $\fE$, $\fA : \pdom \to \R^{\dimx\times \dimx}$, $\fB: \pdom\to \R^{\dimx\times \dimu}$, $\fC: \pdom \to \R^{\dimy\times \dimx}$, and $\fD:\pdom \to \R^{\dimy\times\dimu}$ are matrix valued functions.
We assume throughout this work that $\fE(p)$ is nonsingular for all $p \in \pdom$ (we will briefly discuss extensions to singular $\fE(p)$ in~\Cref{sec:conclusion}) and that the system is \emph{asymptotically stable}, i.\,e., for all $p\in\pdom$, the pencil $(\fE(p), \fA(p))$ only has eigenvalues with negative real part. 
For input-state-output systems~\eqref{eq:fullmodel}, PMOR usually aims at finding ROMs with smaller state dimension $r\ll n_x$ of the form
\begin{align}
\label{eq:redmodel}
\Sigma_r(p):
\begin{cases}
\rE(p) \rdx(t, p) = \rA(p) \rx(t, p) + \rB(p)u(t), \\
\phantom{\rE(p)}\ry(t, p) = \rC(p) \rx(t, p) + \rD(p)u(t),
\end{cases}
\end{align}
with $\rE$, $\rA : \pdom \to \R^{\dimr \times \dimr}$, $\rB: \pdom\rightarrow \R^{\dimr\times \dimu}$, $\rC: \pdom \to \R^{\dimy\times \dimr}$ and~${\rD:\pdom \to \R^{\dimy\times\dimu}}$ such that $y_r \approx y$ for all admissible $u$ and all parameter configurations $p\in \pdom$.

A large body of research has been conducted in the last 20 years in the field of PMOR, ranging from early approaches such as~\cite{Daniel2004, Weile1999}, which are mostly based on multivariate moment matching, to unified frameworks for PMOR that allow for the combination of different MOR strategies to compute parametric ROMs such as~\cite{BaurBBG2011, GeussPL2013}. Another line of research is considered with (greedy) reduced basis methods for computing reduced models for parametric systems \cite{CohenDDN2020,Lassila2012}. In particular in the 2010s, PMOR has evolved into an advanced research field with the development of several sophisticated methods that were used accross different scientific disciplines and industrial applications. 
We refer the reader to~\cite{AntoulasBG2020,BaurBBG2011, BennerGW2015, BennerOCW2017,HesthavenRS2016} for extensive literature reviews and comparisons of different methods.
Most general-purpose state-of-the-art methods perform PMOR by first computing several nonparametric ROMs of the given FOM evaluated at a set of parameter samples and then combining these nonparametric ROMs into one parametric ROM, either by merging the projection subspaces (leading to global-basis methods such as~\cite{BaurBBG2011}) or by interpolating the ROM system matrices~\cite{AmsallemCCF2009, AmsF2011, DegrooteVW2010, GeussPL2013, PanzerML2010} or transfer functions~\cite{BaurB2009}. We will explain these methods and also approximation error measures in more detail in~\Cref{sec:prelim}. For now we only highlight one potential systematic problem with this strategy.

In \cref{fig:problemstatement}, we show the accuracy of ROMs obtained when using the interpolation-based method proposed in~\cite{GeussPL2013} for a Timoshenko beam model depending on a scalar parameter describing the beam length (for details on the system and experimental setup, see~\Cref{sec:numerics}). We first use six interpolation points distributed uniformly in the parameter interval $[0.4, 2.4]$.  Following~\cite{GeussPL2013}, we compute local ROMs at these interpolation points using balanced truncation (BT) consistent to our objective of achieving a small $\hinf$ error. The matrices of these local ROMs are then globalized over the parameter domain by means of piecewise linear interpolation.

 While the error is low at the sample points, it increases drastically (sometimes by more than two orders of magnitude) between the sample points. Even if we almost double the number of sample points (and thus increase the complexity and storage requirement for the ROM), the error still strongly increases between the sample points. This is because the inter-sample-behavior of the FOM is not considered when merging the individual ROMs to one parametric ROM\@ and the sampling does not cover the parameter domain well enough. Conversely, when the parametric ROM is obtained with the parametric extension to SOBMOR presented in this article, the error stays small across the entire parameter range because the FOM behavior can be considered at a large number of parameter samples without increasing the ROM complexity. In fact, the ROM complexity of the SOBMOR-ROM is similar to the complexity of the interpolation-based ROM that uses six interpolation points.

The model in our example depicted in \cref{fig:problemstatement} is particularly challenging for PMOR because variations of the parameter drastically change the model behavior.
We do not want to distort the reader's general impression of current PMOR methods and mention that state-of-the-art PMOR often leads to sufficiently accurate results even when highly complex FOMs are considered as demonstrated in~\cite{AmsF2011, Bui2008, DegrooteVW2010}. A greedy sampling strategy, as proposed in~\cite{Prud2002} may also lead to a more appropriate sample point distribution, which may further decrease the error in state-of-the-art methods. However, \cref{fig:problemstatement} still emphasizes a structural problem that can occur with the popular sample-and-merge approach to PMOR\@.

\begin{figure}[tb]
	\centering
	\input{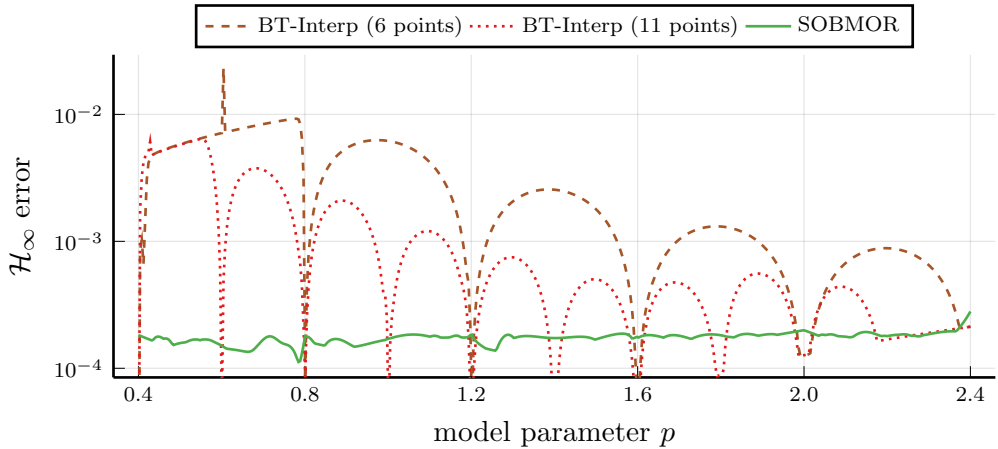}
	\caption{Comparison of a matrix interpolation-based method and SOBMOR in view of the approximation error on the parameter domain.}\label{fig:problemstatement}
\end{figure}


This article describes an optimization-based approach to PMOR, which is based on first making an ansatz for a parametric ROM, and then optimizing the free model parameters to reduce the $\mathcal{H}_\infty \otimes \mathcal{L}_\infty $ approximation error. For this, we extend the parameterization proposed in~\cite{SchwerdtnerV2020} to parametric ROMs and extend the adaptive sampling procedure developed in~\cite{Schwerdtner2021} to the multi-dimensional setting.

The main features of our approach are the high accuracy accross the entire considered parameter domain and the flexibility, as we do not pose any condition on the parametric dependency of the FOM but are still able to preserve structural features of interest such as stability (which is not always guaranteed e.\,g.\ in projection-based PMOR\@; see~\cite{BaurBBG2011}) or a port-Hamiltonian (pH) structure. Structure-preserving MOR for pH systems is currently intensely studied (see~\cite{BeaGM21, BreU21, MehU22, Moser2020, Sato2018} for a few of the most recent articles) due to the benefits of the pH structure such as its modularity and inherent passivity. We explain in \Cref{sec:PHMSD} how our proposed approach can ensure structure-preservation of the pH structure accross the entire parameter range.


Another optimization-based approach to PMOR is introduced in~\cite{HundMMS2021}, in which a method for $\mathcal{H}_2 \otimes \mathcal{L}_2$ optimal PMOR is described. Besides considering a different approximation metric (and optimization strategy) we use structured matrices to ensure ROM stability, whereas method in~\cite{HundMMS2021} uses constrained optimization. We compare our presented approach to the method of~\cite{HundMMS2021} in \Cref{sec:PHMSD}.

Our paper is organized as follows: The next section briefly recalls accuracy measures for PMOR of dynamical systems and explains some existing PMOR methods, as well as structure-preserving MOR\@. In~\Cref{sec:ourmethod} we present our approach and explain our parameterization, optimization, and adaptive sampling strategy. In~\Cref{sec:numerics}, we compare our method to other PMOR methods, which emphasizes its high accuracy and wide applicability. We conclude the paper in~\Cref{sec:conclusion}, where we provide perspectives for future research.

\section{Preliminaries}\label{sec:prelim}
This section provides a background for MOR of parametric systems. We first recall the error measures for linear time-invariant parametric systems following the presentation in~\cite{BennerGW2015} and then review existing PMOR methods that are later used for a comparison in our numerical experiments. Finally, we recall structure-preserving MOR and discuss its extension to PMOR\@.

\subsection{Error Measures for PMOR}

For linear time-invariant dynamical systems such as~\eqref{eq:fullmodel} and their reduced surrogate~\eqref{eq:redmodel}, the output approximation errors
\begin{align}
  \label{eq:l2err}
  \|y( \cdot, p) - \red{y}( \cdot, p)\|_{\ltwo} &:= {\left( \int_{0}^\infty \|y(t, p)-\red{y}(t, p)\|_2^2 \,\, \dd t \right)}^{1/2} \text{ and} \\
  \label{eq:linferr}
  \|y( \cdot, p) - \red{y}( \cdot, p)\|_{\linf} &:=\sup\limits_{t \ge 0} \|y(t, p) - \red{y}(t, p) \|_\infty
\end{align}
can be estimated by comparing the \emph{transfer functions} of FOM and ROM for a fixed parameter value $p\in \Omega$. For homogeneous initial conditions (i.\,e., $x(0, p) = 0$ and $\rx(0, p) = 0$), the transfer functions of~\eqref{eq:fullmodel} and~\eqref{eq:redmodel} are given by
\begin{align}
  \label{eq:transferfunc}
  \tf(s,p) &= \fC(p){(s\fE(p)-\fA(p))}^{-1}\fB(p) + \fD(p) \text{ and} \\
  \tfr(s,p) &= \rC(p){(s\rE(p)-\rA(p))}^{-1}\rB(p) + \rD(p).
\end{align}
These matrix valued functions $\tf(s,p)$ and $\tfr(s,p)$ constitute a mapping of a Laplace-transformed input $U(s)$ to Laplace-transformed outputs $Y(s, p)$ and ${\red{Y}(s, p)}$ via
\begin{align*}
  Y(s, p) = \tf(s, p)U(s) \quad \text{and} \quad \red{Y}(s, p) = \tfr(s, p)U(s).
\end{align*}
Consequently, if the transfer function $\tf(s,p)$ is well-approximated by $\tfr(s,p)$ for all $s\in \mathbb{C}$ with $\Real{s}\geq 0$ and $p\in \Omega$, the Laplace-transformed error $Y(s,p)-\red{Y}(s, p)$ is small. For fixed $p\in \Omega$, the transfer function errors
\begin{align}
  \label{eq:h2err}
  \|\tf( \cdot, p)-\tfr( \cdot, p)\|_{\htwo} &:= {\left(\frac{1}{2\pi}\int_{-\infty}^{\infty} \| \tf(\ri \omega, p) - \tfr(\ri \omega, p)\|_F^2\,\,\dd\omega\right)}^{1/2}, \\
  \label{eq:hinferr}
  \|\tf( \cdot, p)-\tfr( \cdot, p)\|_{\hinf} &:= \sup\limits_{\omega \in \R} \|\tf(\ri \omega, p) -\tfr(\ri \omega, p)\|_2,
\end{align}
where $\| \cdot \|_F$ and $\|  \cdot \|_2$ denote the Frobenius norm and spectral norm, respectively, directly yield upper bounds for the approximation errors~\eqref{eq:l2err} and~\eqref{eq:linferr} given by
\begin{align*}
  \|y(\cdot,p) - \ry(\cdot,p)\|_{\linf} &\leq \|\tf(\cdot,p) - \tfr(\cdot,p)\|_{\htwo}\|u\|_{\ltwo} \text{ and} \\
  \|y(\cdot,p) - \ry(\cdot,p)\|_{\ltwo} &\leq \|\tf(\cdot,p) - \tfr(\cdot,p)\|_{\hinf}\|u\|_{\ltwo},
\end{align*}
respectively. We refer to~\cite{Antoulas2005} for a derivation and detailed analysis. In (nonparametric) MOR, these bounds are usually considered in the context of algorithms minimizing either the $\htwo$ error or the $\hinf$ error. For composite error measures in PMOR, both the frequency and parameter space must be considered \cite{BennerGW2015}.
One such composite error is the~$\htwoltwo$ error
\begin{align*}
  \|\tf - \tfr\|_{\htwoltwo} := \left(\frac{1}{2\pi} \int_{-\infty}^{\infty}\int_{\pdom}\|\tf(\ri \omega, p) - \tfr(\ri \omega, p)\|_F^2\, \dd p\, \dd \omega \right)^{1/2},
\end{align*}
which was introduced in~\cite{BaurBBG2011}. The complementary composite error measure proposed in~\cite{BennerGW2015} is the~$\hinflinf$ error, defined by
\begin{align*}
  \|\tf - \tfr\|_{\hinflinf} := \sup\limits_{\omega \in \R}\,\max\limits_{p \in \pdom}\|\tf(\ri \omega, p) - \tfr(\ri \omega, p)\|_2.
\end{align*}
Nonparametric MOR methods that lead to a good or even optimal performance in either the $\hinf$ or the $\htwo$ error are well studied. Methods for pMOR targeting the $\htwoltwo$ error have been suggested in~\cite{BaurBBG2011} for a specific parameter dependency and for general parametric systems, an optimization-based strategy was proposed in~\cite{HundMMS2021}. In this article, we propose a PMOR method, which yields small $\hinflinf$ errors\@.
 
\subsection{Review of PMOR methods}\label{subsec:existing}

In this part, we summarize PMOR methods for linear time-invariant systems that we also use in our comparison in \Cref{sec:numerics}. We discuss the strategy behind both global and local reduction approaches. 
In the former, the idea is to obtain a reduced model from projecting the full order parametric model by projection matrices, which are constructed to be suitable for the whole parameter range (and thus are global in the parameter). 

 In contrast to global basis methods, the basis used for computing the parametric ROM in a local method is usually obtained by interpolating the nonparametric local projection matrices, the ROM matrices, or the ROM transfer functions, all of which are computed by (nonparametric) MOR at a set of sample points $\{p^{(1)}, \dots, p^{(n_s)}\}\subset \Omega$. Interpolation of projection matrices requires access to the FOM for each parameter update, which may lead to increased storage and computational demands (with an exception for systems considered in~\cite{Son2013} or if a further treatment is applied as in~\cite{Wittmuess2016}), and the interpolation of the transfer function as presented in~\cite{BaurB2009} leads to an increased ROM dimension for each parameter sample considered. Therefore, in this article, we study only matrix interpolation methods as in~\cite{AmsallemCCF2009, AmsF2011, DegrooteVW2010, LohmannE2007, PanzerML2010} as a representative of local reduction methods.

In order to keep the presentation concise, we do not go into detail about parameter sampling strategies and refer to~\cite[Section~3.4]{BennerGW2015} for this issue.

\subsubsection{Global basis methods}
\label{subsub:global}
In global basis methods, the ROM takes the form
\begin{align}
  \label{eq:globbasrom}
  \begin{split}
  W^\T \fE(p) V \rdx(t, p) &= W^\T \fA(p) V \rx(t, p) + W^\T \fB(p) u(t), \\
  \ry(t, p) &= \fC(p) V \rx(t, p) + \fD(p)u(t),
\end{split}
\end{align}
where $V, W \in \R^{\dimx \times \dimr}$ is a constant (two-sided) projection basis that is used for all parameter values and is global in the parameter in this sense. If the system matrices have a parameter-separable form
  $A(p) = \sum_{i=1}^\ell f_i(p)A_i$
with scalar parameter dependencies $f_i:\Omega \to \mathbb{R}$, $i = 1, \dots, \ell$, and constant matrix coefficients $A_i \in \R^{\dimx \times  \dimx}$, $i = 1, \dots, \ell$, and if $\ell \in \mathbb{N} $ is of moderate size, the ROM can be evaluated efficiently, since the parameter-independent reduced order coefficients $W^\T A_i V$, $i = 1, \dots, \ell$, can be precomputed. 
For a more complex parameter dependency, where $\ell \gg 1$, 
an efficient evaluation of the ROM could be achieved by means of (discrete) empirical interpolation methods~\cite{BarraulsMNP2004,ChaturantabutS2010}.

The projection matrices $V$ and $W$ are typically computed as follows. First, a set of sample points $\{p^{(1)}, \dots, p^{(n_s)}\}\subset \Omega$ is chosen, for which a set of projection matrices $V_i$ and $W_i$ for $i = 1, \dots, n_s$ are computed by performing nonparametric projection-based MOR (e.\,g.\ balanced truncation~\cite{Moore1981, MullisR1976} or the iterative rational Krylov algorithm (IRKA)~\cite{GugercinA2008}). Then, global projection matrices are constructed via $V := [V_1, \dots, V_{n_s}]$, and $W := [W_1, \dots W_{n_s}]$. Usually, the individual projection matrices $V_i$ and $W_i$, $i=1,\ldots,n_s$ are orthogonalized prior to concatenation and the final matrices $V$ and $W$ are again treated by a singular value decomposition to remove any (nearly) rank-deficient parts. For greedy allocation of the sample parameter points $\{p^{(1)}, \dots, p^{(n_s)}\}$, we refer to, e.g.,~\cite{Binder2021} and~\cite{Prud2002}.

It is important to note that in the global basis case, the structural properties of the nonparametric MOR methods that are used to construct the individual projection matrices do not carry over to the parametric ROM\@. In particular, the parametric ROM is neither balanced nor guaranteed to be stable, when using BT to compute the local projections $V_i$ and $W_i$, $i=1,\ldots,n_s$.
In case of IRKA, which is based on interpolation of the transfer function, the interpolation conditions at the sample points $\{p^{(1)}, \dots, p^{(n_s)}\}$ are retained. However, the $\htwo$-optimality of the local model, which is achieved by IRKA, can only be extended to an $\htwoltwo$-optimality of the parametric model in particular cases, such as parameter dependencies only present in input and output matrices~\cite{BaurBBG2011, KleymanG2020}. Furthermore, the stability of the ROM, which is often an essential property, cannot be ensured in general using global basis methods. Note that a recent development~\cite{Gosea2021} allows for retaining interpolation conditions along all parameters $p \in \pdom$. However, the method involves parameter dependent projection matrices, which requires particular effort in the precomputation of the ROM matrices.

The preservation of stability in global basis methods can be ensured in particular cases by using structured approximations, e.g.\ by means of one-sided projections (choosing $V=W$) applied to systems with negative definite system pencils ($\fE(p), \fA(p))$, that is, for all $p \in \pdom$, the matrices $\fE(p)$ and $\fA(p)$ are positive definite and negative definite, respectively, cf.~\cite{BennerGW2015}. An extension to BT, which ensures stability for general stable FOMs and even provides balanced systems for all $p \in \pdom$ requires parameter dependent projection matrices~\cite{Wittmuess2016}. For a discussion of general optimal parametric approximations using Kolmogorov n-widths, we refer the reader to~\cite{UngerG2019}.

\subsubsection{Matrix interpolation methods}\label{subsub:local}

In~\cite{GeussPL2013}, a template that explains the different choices for constructing the interpolated ROM at each stage of matrix interpolation based PMOR is provided. In matrix interpolation methods, first, initial ROMs are computed using standard nonparametric MOR using samples of the FOM at a sample set $\{p^{(1)}, \dots, p^{(n_s)}\}\subset \Omega$. After that, the different ROMs are aligned, which ensures that the ROMs are described in the same set of generalized coordinates (this does not change transfer functions of the individual ROMs at $p^{(i)}$).
Then the ROM matrices are interpolated to form a parametric ROM, typically on a matrix manifold because the system matrix may be restricted to lie on the manifold of regular matrices; see~\cite{DegrooteVW2010, GeussPL2013} for details.

\subsection{Structure-Preserving MOR for pH systems}\label{sec:phreduction}

One branch of MOR considers the preservation of structural features of the original system during the reduction, which may include a second order structure with symmetric positive semi-definite coefficient matrices, passivity, or the pH structure (which ensures passivity). We restrict our presentation to the pH structure as this has recently become an important building block in the modeling and control of multi-physical and network systems; see~\cite{MehU22} for a recent survey demonstrating the wide applicability and the benefits of pH systems.
In view of the invertibility of the descriptor matrix in \eqref{eq:fullmodel}, we consider the case $E(p) = I$ in what follows. For a definition of pH-systems with singular descriptor matrix, we refer to \cite{MehU22}.
\begin{definition}[Parametric port-Hamiltonian Systems]
A linear time-invariant parametric system
\begin{align*}
  \dot x(t, p) &= (\fJ(p)-\fR(p))\fQ(p) x(t, p) + (\fG(p) - \fP(p))u(t), \\
       y(t, p) &= {(\fG(p) + \fP(p))}^\T\fQ(p) x(t, p) + (\fS(p) + \fN(p))u(t),
\end{align*}
where $\fJ, \fR, \fQ : \pdom  \rightarrow \R^{\dimx \times \dimx}$, $\fG, \fP: \pdom  \rightarrow \R^{\dimx \times \dimu}$, and $\fS, \fN: \pdom  \rightarrow \R^{\dimu \times \dimu}$, is called a \emph{parametric port-Hamiltonian system}, if the following conditions are satisfied for all $p \in \pdom$:
\begin{enumerate}[(i)]
  \item the matrices $\fJ(p)$ and $\fN(p)$ are skew-symmetric,
  \item the \emph{passivity matrix} $ W_{\rm pas}(p) := \begin{bsmallmatrix} \fQ{(p)}^\T \fR(p) \fQ(p) & \fQ(p) \fP(p) \\ \ {(\fQ{(p)} \fP{(p)})}^\T & \fS(p) \end{bsmallmatrix} $
    and $Q(p)$ are symmetric positive semi-definite.
\end{enumerate}\label{def:phsys}
\end{definition}
We note that $Q$ is often set to the identity by introducing a descriptor matrix; see~\cite[Section 4.3]{MehU22}.
The structure-preserving MOR of pH systems is concerned with finding a low-order approximation of a large-scale parametric pH system that also satisfies the conditions in~\cref{def:phsys}. For a structure-preserving PMOR it must be ensured that the ROM is pH for all parameter configurations $p \in \pdom$.  
We explain how the pH structure can be encoded in our proposed parameterization in \Cref{sec:PHMSD}.

There exist several methods for nonparametric structure-preserving MOR of pH systems (see, e.\,g.,~\cite{BeaGM21, BreU21, Gugercin2012, HauMM19, PolyugaS2010}). The development of structure-preserving MOR algorithms for parametric pH systems has only recently started~\cite{StahlSM2022}. In order to be able to compare the method proposed in this article to another structure-preserving method, we extend IRKA-PH, which is the pH structure-preserving variant of IRKA, to the parametric case using the framework of~\cite{BaurBBG2011}. This generalization is straightforward in case $\fQ$ is constant. For a set of sample-points $\{p^{(1)}, \dots, p^{(n_s)}\}$, we can compute a set of projection matrices $V_i$, $i=1,\ldots,n_s$ using IRKA-PH~\cite{Gugercin2012}. These are concatenated to one global projection matrix $V = [V_1, \dots, V_{n_s}]$. As in nonparametric IRKA-PH, we then compute the projection $W := QV{\left(V^\T Q V\right)}^{-1}$ and project the system via
\begin{align*}
  \rdx(t,p) &= (W^\T \fJ(p) W-W^\T \fR(p) W)V^\T \fQ V \rx(t,p) + W^\T (\fG(p) - \fP(p))u(t), \\
  \ry(t,p) &= \left({(\fG(p) + \fP(p))}^\T W \right) V^\T \fQ V \rx(t,p) + (\fS(p) + \fN(p))u(t).
\end{align*}
This projection retains the symmetry and definiteness properties as imposed by \cref{def:phsys} for all $p\in\pdom$ and is therefore structure-preserving.


\section{Our Method}\label{sec:ourmethod}

We present our optimization-based PMOR method in the following order: First, we introduce our parameterization of a general parameter-dependent ROM\@. Then we explain how the free parameters of the ROM can be optimized to obtain a small $\hinflinf$ error. For this, we provide an optimization strategy in the spirit of~\cite{SchwerdtnerV2020} that is based on transfer function evaluations. Moreover, we generalize the adaptive frequency sampling strategy~\cite{Schwerdtner2021} to the parametric case for an efficient adaptive optimization-based ROM computation.

\subsection{Parameterization}

We extend the approach presented in~\cite{SchwerdtnerV2020} in order to parameterize a general class of stable parametric systems. For this, we distinguish between the \emph{design} parameter vector $\theta \in \R^{\dimtheta}$, $\dimtheta\in \mathbb{N}$, which can be used to tune the ROM during the optimization (and remains fixed thereafter), and the model parameter vector $p \in \pdom$. Our parameterized ROMs have the form
\begin{align}
  \label{eq:romform}
  \Sigma(p, \theta):
  \begin{cases}
    \dot x = \pA(p, \theta) x + \pB(p, \theta), \\
         y = \pC(p, \theta) x + \pD(p, \theta),
  \end{cases}
\end{align}
where
$\pA: \pdom \myotimes \R^{\dimtheta}  \rightarrow \R^{\dimr \times \dimr}$,
$\pB: \pdom \myotimes \R^{\dimtheta}  \rightarrow \R^{\dimr \times \dimu}$,
$\pC: \pdom \myotimes \R^{\dimtheta}  \rightarrow \R^{\dimy \times \dimr}$, and
$\pD: \pdom \myotimes \R^{\dimtheta}  \rightarrow \R^{\dimy \times \dimu}$.
The transfer function of $\Sigma(p, \theta)$ is given by
\begin{align}
  \label{eq:paramTF}
\ptf(s, p; \theta) = \pC(p, \theta){\left(sI-\pA(p, \theta)\right)}^{-1}\pB(p, \theta)+\pD(p, \theta).
\end{align}
Analogously to~\cite{SchwerdtnerV2020}, we optimize $\theta$ to obtain a transfer function~$\ptf( \cdot,  \cdot, \theta)$ that approximates a given transfer function $\tf(\cdot,\cdot)$.

In the reduced model, we must ensure that the matrix $\pA( p, \theta)$ is an asymptotically stable matrix (i.\,e., $\pA(p, \theta)$ only has eigenvalues with strictly negative real part) for all admissible parameters $p \in \pdom$ and ${\theta \in \R^{\dimtheta}}$. In order to avoid constrained optimization, we pursue a straight-forward parameterization, where we exploit the equivalence of Dissipative-Hamiltonian (DH) matrices and stable matrices~\cite{GillisS2017}. 
\begin{definition}[Dissipative-Hamiltonian matrix~\cite{GillisS2017}]\label{def:DH} 
  A matrix $M \in \R^{n \times  n}$ is Dissipative-Hamiltonian (DH) if $M = (J-R)Q$ for some $J, R, Q \in \R^{n \times n}$, where $J$ is skew-symmetric, $R$ is symmetric positive semi-definite, and $Q$ is symmetric positive definite.
\end{definition}
We can thus parameterize a stable matrix as a composition of parameterized skew-symmetric and symmetric positive (semi)-definite matrices. In particular, we set
\[ \pA(p, \theta) = (\pJ(p, \theta)-\pR(p, \theta))\pQ(p, \theta), \]
where 
$\pJ: \pdom \myotimes \R^{\dimtheta}  \rightarrow \R^{\dimr \times \dimr}$ takes values in the set of skew-symmetric matrices, while
$\pR: \pdom \myotimes \R^{\dimtheta}  \rightarrow \R^{\dimr \times \dimr}$ and 
$\pQ: \pdom \myotimes \R^{\dimtheta}  \rightarrow \R^{\dimr \times \dimr}$ both take values in the set of symmetric positive semi-definite matrices. Several remarks are in order.
\begin{remark}[Ensuring asymptotic stability]\label{rem:PSD} \hfill

	\begin{enumerate}
\item  Whereas the matrix $Q$ in \cref{def:DH} is positive definite, our parameterization will only ensure the positive \emph{semi}-definiteness of $\pQ(p, \theta)$. The positive definiteness can be guaranteed by using $\widetilde \pQ(p, \theta) := \pQ(p, \theta) + \varepsilon I$, where $\varepsilon>0$ is a small positive scalar, and subsequently using $\widetilde \pQ$ during the optimization. 
\item Furthermore, using matrices as in \cref{def:DH} only ensures stability, while we require asymptotically stable ROMs.
In our numerical experiments this has not led to any problems (all our ROMs are asymptotically stable for parameter values $p \in \pdom$ that we have tested). 
If we want to explicitly ensure asymptotic stability, we could proceed similar as above for $\pQ$ and set $\widetilde\pR(p,\theta):=\pR(p,\theta)+\varepsilon I$ and thus ensure an asymptotically ROM.
\item  Using this DH structure to ensure stability is preferable especially for PMOR, since the alternative (using a general system matrix $\pA(p, \theta)$ and imposing stability as a constraint) requires the solution of a constrained optimization problem, where the constraint is given by
$$
\max\limits_{p \in \pdom, i \in \{1, \dots, r\}}\Real\Lambda_i(\pA(p, \theta)),
$$
where $\Lambda_i$ denotes the $i$-th eigenvalue of its matrix argument. This quantity is not only hard to compute but it also depends nonsmoothly on $\theta$, which poses additional challenges during optimization.
	\end{enumerate}
\end{remark}

\begin{remark}
  Our parameterization~\eqref{eq:romform} does not include a descriptor matrix. This is because we have assumed that $\fE(p)$ in~\eqref{eq:fullmodel} is nonsingular, implying that~\eqref{eq:fullmodel} is an ordinary differential equation which has a proper transfer function with $\lim\limits_{s  \rightarrow \infty } \tf(s, p) = D(p)$. This can be approximated by means of our ansatz.
  In \Cref{sec:conclusion}, we comment on extensions of our approach to FOMs with a singular $\fE(p)$, which may have improper transfer functions.
\end{remark}

\noindent For the matrix-valued functions $\pB, \pC, \pD, \pJ, \pR$, and $\pQ$ we employ the ansatz
\begin{align}
  \label{eq:ansatz}
  \begin{split}
  \mathcal{B}(p, \theta) = \sum\limits_{i=1}^{\kappa_B} f_i^B(p) B_i(\theta), \quad
  \mathcal{C}(p, \theta) = \sum\limits_{i=1}^{\kappa_C} f_i^C(p) C_i(\theta), \\
  \mathcal{D}(p, \theta) = \sum\limits_{i=1}^{\kappa_D} f_i^D(p) D_i(\theta), \quad
  \mathcal{J}(p, \theta) = \sum\limits_{i=1}^{\kappa_J} f_i^J(p) J_i(\theta), \\
  \mathcal{R}(p, \theta) = \sum\limits_{i=1}^{\kappa_R} f_i^R(p) R_i(\theta), \quad
  \mathcal{Q}(p, \theta) = \sum\limits_{i=1}^{\kappa_Q} f_i^Q(p) Q_i(\theta), \\
  \end{split}
\end{align}
where the functions $\funB_1, \dots, \funB_{\kappaB}: \Omega  \rightarrow \R$, $\funC_1, \dots, \funC_{\kappaC}: \Omega  \rightarrow \R$, $\funD_1, \dots, \funD_{\kappaD}: \Omega  \rightarrow \R$, $\funJ_1, \dots, \funJ_{\kappaJ}: \Omega  \rightarrow \R$, $\funR_1, \dots, \funR_{\kappaR}: \Omega  \rightarrow \R$, and~${\funQ_1, \dots, \funQ_{\kappaQ}: \pdom  \rightarrow \R}$ capture the dependency on the model parameter. 

In \cref{rem:stabAllp}, we discuss how the positive semi-definiteness and symmetry of $\pR(p, \theta)$ and $\pQ(p, \theta)$ and the skew-symmetry of $\pJ(p,\theta)$ defined in \eqref{eq:ansatz} can be ensured for all $p$ and $\theta$. Before that, we present reshaping operations, which constitute the building blocks of the $\theta$-dependent parts in the system matrix functions.
\begin{definition}[Reshaping operations as in~\cite{SchwerdtnerV2020}]
  \begin{enumerate}[a)] 
    \item The function family  
      \begin{align*}
        \vtf_m: \C^{n \cdot m} \rightarrow \C^{n\times m}, \quad  v \mapsto
        \begin{bmatrix}
          v_1 & v_{n+1} & \dots & v_{m(n-1)+1}\\
          v_2 & v_{n+2} & \dots & v_{m(n-1)+2}\\
          \vdots & \vdots &  & \vdots \\
          v_n & v_{2n} & \dots & v_{nm}
        \end{bmatrix}
      \end{align*}
      reshapes a vector into an accordingly sized matrix with $m$ columns. Here, $\vtf$ stands for vector-to-full (matrix).
      The inverse operation is given by
      \begin{align*}
        \ftv: \C^{n \times m}  \rightarrow \C^{n \cdot m}, \quad A \mapsto
        \begin{bmatrix}
          a_{1,1} &a_{2,1} &\dots & a_{n,1} & a_{1,2} & \dots & a_{n,m}
        \end{bmatrix}^\T,
      \end{align*}
      where $\ftv$ stands for full (matrix)-to-vector given by the standard vectorization operator usually denoted by $\operatorname{vec}$.
    \item The function 
      \begin{align*}
        \vtu : \C^{n(n+1)/2}  \rightarrow \C^{n\times n}, \quad v \mapsto
        \begin{bmatrix}
          v_1 & v_2 & \dots & v_n \\
          0 & v_{n+1}& \dots & v_{2n-1} \\
          \vdots & \vdots & \ddots & \vdots \\
          0 & 0 & \dots & v_{n(n+1)/2} \\
        \end{bmatrix}
      \end{align*}
      maps a vector of length $n(n+1)/2$ to an $n\times n$ upper triangular 
      matrix (where $\vtu$ stands for vector-to-upper (triangular)), while the function 
      \begin{align*}
        \utv : \C^{n\times n}  \rightarrow \C^{n(n+1)/2}, \quad A \mapsto
        \begin{bmatrix}
          a_{1,1} & a_{1,2} & \dots & a_{1,n} & a_{2,2} & \dots & a_{n,n}
        \end{bmatrix}^\T
      \end{align*}
      maps the upper triangular part of an $n\times n$ matrix to a vector (where $\utv$ stands for upper (triangular)-to-vector).
    \item The function
      \begin{align*}
        \vtsu : \C^{n(n-1)/2}  \rightarrow \C^{n\times n}, \quad v \mapsto
        \begin{bmatrix}
          0 & v_1 & v_2 & \dots  & v_{n-1} \\
          0 & 0   & v_n & \dots  & v_{2n-3}  \\
          \vdots & \vdots  & \vdots   & \ddots & \vdots \\
          0 & 0   & 0   &  \dots     & v_{n(n-1)/2} \\
          0 & 0   & 0   & \dots      & 0\\
        \end{bmatrix}
      \end{align*}
      maps a vector of length $n(n-1)/2$ to an $n\times n$ strictly upper 
      triangular matrix (where $\vtu$ stands for vector-to-strictly upper (triangular)), while the function 
      \begin{align*}
        \sutv : \C^{n \times n}  \rightarrow \C^{n(n-1)/2}, \quad A \mapsto
        \begin{bmatrix}
          a_{1,2} & a_{1,3} & \dots & a_{1,n} & a_{2,3} & \dots & a_{n-1,n}
        \end{bmatrix}^\T
      \end{align*}
       maps the strictly upper triangular part of an $n\times n$ matrix to a 
       vector (where $\sutv$ means strictly upper (triangular)-to-vector).
  \end{enumerate}
\end{definition}

\noindent Utilizing these functions, the $\theta$-dependent parts of $\pB, \pC, \pD, \pJ, \pR$, and $\pQ$ are defined by
\begin{align*}
  B_i(\theta) &:= \vtf_{\dimu}(\theta_{B_i})                      \quad &&\text{ for } i \in \{1, \dots, \kappaB \},\\
  C_i(\theta) &:= \vtf_{\dimy}(\theta_{C_i})                      \quad &&\text{ for } i \in \{1, \dots, \kappaC \},\\
  D_i(\theta) &:= \vtf_{\dimy}(\theta_{D_i})                      \quad &&\text{ for } i \in \{1, \dots, \kappaD \},\\
  J_i(\theta) &:= \vtsu(\theta_{J_i}) - \vtsu{(\theta_{J_i})}^\T  \quad &&\text{ for } i \in \{1, \dots, \kappaJ \},\\
  R_i(\theta) &:= \vtu(\theta_{R_i}) \vtu{(\theta_{R_i})}^\T      \quad &&\text{ for } i \in \{1, \dots, \kappaR \},\\
  Q_i(\theta) &:= \vtu(\theta_{Q_i}) \vtu{(\theta_{Q_i})}^\T      \quad &&\text{ for } i \in \{1, \dots, \kappaQ \}.
\end{align*}
Here the (design) parameter vector $\theta$ is partitioned as
\begin{align}
\label{eq:partitioning}
\begin{split}
  \theta = [\theta_{B_1}^\T, \dots, \theta_{B_{\kappaB}}^\T, \theta_{C_1}^\T, \dots, \theta_{C_{\kappaC}}^\T,
  \theta_{D_1}^\T, \dots,& \theta_{D_{\kappaD}}^\T, \theta_{J_1}^\T, \dots, \theta_{J_{\kappaJ}}^\T,\\
         & \pushright{\theta_{R_1}^\T, \dots, \theta_{R_{\kappaR}}^\T, \theta_{Q_1}^\T, \dots, \theta_{Q_{\kappaQ}}^\T
  ]^\T,}
\end{split}
\end{align}
with
$\theta_{B_i}~\in~\R^{\dimr\cdot \dimu}$, 
$\theta_{C_i}~\in~\R^{\dimy\cdot \dimr}$, 
$\theta_{D_i}~\in~\R^{\dimy\cdot \dimu}$, 
$\theta_{J_i}~\in~\R^{\dimr\cdot(\dimr-1)/2}$, 
$\theta_{Q_i}~\in~\R^{\dimr\cdot(\dimr+1)/2}$, and 
$\theta_{R_i}~\in~\R^{\dimr\cdot(\dimr+1)/2}$. By this parameterization, skew-symmetry of $J_i(\theta)$  and symmetry and positive semi-definiteness of all $Q_i(\theta)$  and $R_i(\theta)$ follows straightforwardly.

\begin{remark}\label{rem:stabAllp}
  If $\pR(p, \theta)$ and $\pQ(p, \theta)$ are defined as in~\eqref{eq:ansatz}, then we can only use functions $f_1^R, \dots, f_{\kappaR}^R$ and $f_1^Q, \dots, f_{\kappaQ}^Q$ that take nonnegative values, since otherwise, we cannot ensure the positive semi-definiteness of the summands of $\pR(p, \theta)$ and $\pQ(p, \theta)$. In our implementation, we use shifted hat-functions (see~\eqref{eq:hat}), which only attain nonnegative values.

  However, we can modify $\pR$ and $\pQ$ to also allow for scalar functions, which may attain negative values, as follows. For this, let
  \begin{align*}
    \mathcal{V}_R(p, \theta) = \sum\limits_{i=1}^{\kappaR}f^R_i(p)\vtu(\theta_{R_i}), \quad 
    \mathcal{V}_Q(p, \theta) = \sum\limits_{i=1}^{\kappaQ}f^Q_i(p)\vtu(\theta_{Q_i}).
  \end{align*}
  Then we can define $\mathcal{R}_{r,\text{mod}}$ and $\mathcal{Q}_{r,\text{mod}}$ via
  \begin{align*}
    \mathcal{R}_{r,\text{mod}}(p, \theta) = \mathcal{V}_R(p, \theta) \mathcal{V}_R{(p, \theta)}^\T, \quad 
    \mathcal{Q}_{r,\text{mod}}(p, \theta) = \mathcal{V}_Q(p, \theta) \mathcal{V}_Q{(p, \theta)}^\T.
  \end{align*}
  This ensures that $\mathcal{R}_{r,\text{mod}}(p, \theta)$ and $\mathcal{Q}_{r,\text{mod}}(p, \theta)$ are positive semi-definite for all $p \in \pdom$ for all values of $f_1^R, \dots, f_{\kappaR}^R$ and $f_1^Q, \dots, f_{\kappaQ}^Q$.
\end{remark}

\subsection{Optimization}
The computation of the reduced model is performed by tuning $\theta$ to iteratively reduce the error $\|H - \ptf(\cdot, \cdot, \theta)\|_{\hinfhinf}$. 
For that, we do not minimize the $\hinfhinf$ error directly but instead minimize the objective function
\begin{align}
  \begin{split}
  \loss(\theta; \tf,\gamma,\mathcal{S}) := \frac{1}{\gamma}\sum\limits_{(\omega_i, \param_i) \in \mathcal{S}} \sum\limits_{j=1}^{\min(\dimu, \dimy)}{\left({\left[\sigma_j \left(\tf(\ri \omega_i, \param_i)-\ptf(\ri \omega_i, \param_i; \theta)\right)-\gamma\right]}_+\right)}^2
\end{split}
  \label{eq:loss}
\end{align}
with respect to $\theta$ on a sequence of decreasing values of $\gamma > 0$, where $\sigma_j(\cdot)$ denotes the $j$-th singular value of its matrix argument. Further,
\begin{align*}
  {[\, \cdot\, ]}_+:  \R \rightarrow [0,\infty), \quad x \mapsto 
  \begin{cases}
    x & \text{if } x\ge 0,\\
    0 & \text{if } x<0,
  \end{cases}
\end{align*}
denotes the positive part of a scalar and the set $\mathcal{S} \subset \R \times \pdom$ contains the sample points both in frequency and parameter domain, at which both the original and our parameterized transfer function are evaluated.

In the following, we highlight several favorable properties of $\loss$ which motivate the use of this objective functional rather than the $\hinfhinf$ error. For that, we first restate~\cite[Proposition~3.1]{SchwerdtnerV2020} for the case of parametric systems.
We assume that both $\tf$ and $\ptf$ depend smoothly on the parameter vector $p \in \pdom$ and the Laplace variable $s \in \C$.

\begin{proposition}[Properties of $\loss$] Let $\mathcal{S} =\{(\omega_1,p^{(1)}),\,\ldots,\,(\omega_k, p^{(k)}) \} \subset \R \times \R^{\dimp}$, $\theta_0 \in \R^{n_\theta}$, and $\gamma > 0$ be fixed and let $\loss$ be given as in~\eqref{eq:loss}. For $i=1,\,\ldots,\,k$ and $j~=~1,\,\ldots,\,\min(\dimu, \dimy)$ define
\begin{align*}
  f_{ij}(\theta_0) &:= \sigma_j(\tf(\ri \omega_i, p^{(i)})-\ptf(\ri \omega_i, p^{(i)}; \theta_0))\quad \text{and}\\
  g(\theta_0) &:= \loss(\gamma,\tf,\ptf(\cdot, \cdot, \theta_0),\mathcal{S}).
\end{align*}
Then the following statements hold:
  \begin{enumerate}[i)]
    \item We have that $\loss(\gamma,\tf,\ptf(\cdot, \cdot, \theta_0),\mathcal{S})=0$ for all $\gamma>{\|\tf-\ptf(\cdot, \cdot, \theta_0)\|}_{\hinfhinf}$.
    \item 
    The function $g(\cdot)$ is differentiable. Moreover, the partial derivatives of $g(\cdot)$ at $\theta_0$ are given by
    \begin{align}\label{eq:diffg}
     \begin{split}
      \frac{\partial}{\partial \theta_{\ell}}g(\theta_0) &= \frac{2}{\gamma} \sum_{f_{ij}(\theta_0) > \gamma} (f_{ij}(\theta_0) -\gamma) \frac{\partial_+}{\partial \theta_\ell} f_{ij}(\theta_0) \\ &= \frac{2}{\gamma} \sum_{f_{ij}(\theta_0) > \gamma} (f_{ij}(\theta_0) -\gamma) \frac{\partial_-}{\partial \theta_\ell} f_{ij}(\theta_0), \quad \ell=1,\,\ldots,\,n_\theta,
      \end{split}
    \end{align}
    where $\frac{\partial_+}{\partial \theta_\ell}$ and $\frac{\partial_-}{\partial \theta_\ell}$ denote the \emph{right} and \emph{left} partial derivative with respect to $\theta_\ell$ and $\theta_\ell$ denotes the $\ell$-th element of $\theta_0$.
  \end{enumerate}
  \label{prop:lossprops}
\end{proposition}

\begin{proof}
  The proof carries over from the nonparametric case~\cite[Proof of Proposition~3.1]{SchwerdtnerV2020}, as the claim only involves the derivative with respect to~$\theta$.
\end{proof}

\noindent Both properties in \cref{prop:lossprops} promote the use of~$\loss$ in an optimization loop to attain a good $\hinflinf$ fit: Property~(i) establishes a connection between a minimization of~$\loss$ and the reduction of the $\hinflinf$ error, while the differentiability established in Property~(ii) facilitates the numerical optimization of~$\loss$ with gradient-based methods, which is explained in the following remark.
\begin{remark}
  The benefits of a minimization of $\loss$ compared to directly minimizing the $\hinfhinf$ error between the FOM and the ROM directly carry over from~\cite[Remark~3.3]{SchwerdtnerV2020}. Moreover, in the parametric case, the direct minimization of the $\hinflinf$ error poses even greater challenges than the minimization of the $\hinf$ error in the nonparametric case.
  \begin{enumerate}
    \item The $\hinf$ norm computation of a large-scale transfer function (such as the error transfer function between FOM and ROM) to a sufficient accuracy is computationally expensive. There exist several methods for the $\hinf$ norm computation of large-scale (nonparametric) systems such as~\cite{AliyevBMSV2017, Benner2012, Freitag2014, Guglielmi2013, MitO15b}. However, for the computation of the $\hinf\otimes\linf$-error, we are not aware of any fast and reliable algorithms.
    \item Even though the transfer function depends smoothly on the parameter vector and the Laplace variable, the $\hinflinf$ error only depends continuously and not differentiably on the parameter vector~$\theta$. This may obstruct the gradient-based numerical optimization.
    \item When using $\loss$, we include information on the error transfer function at all sample points, at which the error is larger than $\gamma$. In contrast, an evaluation of the $\hinflinf$ error only contains information on the current maximum of error of the transfer functions.
  \end{enumerate}
\end{remark}
Our PMOR method, based on bisection over $\gamma$, is described in \cref{alg:bisection}. In each iteration, after updating $\gamma$, the sample points are updated using our adaptive sampling algorithm described in \Cref{sec:adaptiveSampling}. Then, the objective function $\loss$ is minimized by means of a nonlinear optimization algorithm. We use the Broyden-Fletcher-Goldfarb-Shanno (BFGS) algorithm as implemented in~\cite{MogensenR2018} with default parameters\footnote{see \url{https://julianlsolvers.github.io/Optim.jl/stable/algo/lbfgs/}}. If the optimization leads to an objective value that is lower than a prescribed tolerance~${\varepsilon_2>0}$, that is,~${\mathcal{L}(\theta;H,\gamma,\mathcal{S})<\varepsilon_2}$, 
the solver has managed to reduce the error transfer function at all sample points below the current $\gamma$-level up to the tolerance~$\varepsilon_2$, i.e., $\frac{1}{\gamma}{({ [\sigma_j (\tf(\ri \omega_i, \param_i)-\ptf(\ri \omega_i, \param_i; \theta))-\gamma]}_+)}^2 <\varepsilon_2$ and thus $$\sigma_j \left(\tf(\ri \omega_i, \param_i)-\ptf(\ri \omega_i, \param_i; \theta)\right)^2 < \gamma( \varepsilon_2 +\gamma)$$ for all $(\omega_i,p^{(i)})\in \mathcal{S}$ and $j=1,\ldots,\min(n_u,n_y)$. In this case we reduce the upper bound~$\gamma_{\rm u}$. Otherwise, the lower bound $\gamma_{\rm l}$ is increased. The bisection is terminated, when the relative distance between $\gamma_{\rm u}$ and $\gamma_{\rm l}$ is lower than the bisection tolerance~$\varepsilon_1$.

\begin{algorithm}[tbh]
 \hspace*{\algorithmicindent} \textbf{Input:} FOM transfer function $\tf$, initial ROM transfer function $\ptf( \cdot,  \cdot; \theta_0)$ as in~\eqref{eq:paramTF} with parameter $\theta_0 \in \R^{\dimtheta}$, initial sample point set $\mathcal{S} \subset \mathbb{R}\times \Omega$, upper bound $\gamma_{\rm u} > 0$, bisection tolerance $\varepsilon_1 > 0$, termination tolerance $\varepsilon_2 > 0$ \\
 \hspace*{\algorithmicindent} \textbf{Output:} final ROM parameters $\theta_{\rm fin} \in \R^{n_\theta}$
  \begin{algorithmic}[1]
  \STATE{Set $j:=0$ and $\gamma_{\rm l}:=0$.}
  \WHILE{$(\gamma_{\rm u}-\gamma_{\rm l})/(\gamma_{\rm u}+\gamma_{\rm l}) > \varepsilon_1$}
    \STATE{Set $\gamma=(\gamma_{\rm u}+\gamma_{\rm l})/2$.}
    \STATE{Update sample set $\mathcal{S}$ using~\cref{alg:sampling}.}
    \STATE{Solve the minimization problem $\alpha := \min_{\theta\in \R^{n_\theta}}\loss(\theta;\tf,\gamma,\mathcal{S})$ with minimizer $\theta_{j+1}~\in~\R^{n_\theta}$, initialized at $\theta_j$.}
  \IF{$\alpha > \varepsilon_2$}
    \STATE{Set $\gamma_{\rm l}:=\gamma$.}
    \ELSE
    \STATE{Set $\gamma_{\rm u}:=\gamma$.}
  \ENDIF
  \STATE{Set $j:=j+1$.}
  \ENDWHILE
  \STATE{Set $\theta_{\rm fin}:= \theta_j$.}
  \caption{SOBMOR-$\hinf$}\label{alg:bisection}
  \end{algorithmic}
\end{algorithm}

\subsection{Gradient Computation}

To compute the gradient of~$\loss$ as defined in \eqref{eq:loss} analytically, we derive the gradients of the singular values of the difference of the reduced order transfer function~\eqref{eq:paramTF} and the given transfer function~\eqref{eq:transferfunc} at a fixed sample point with respect to~$\theta$. These computations are a straightforward adaptation of the non-parametric case \cite{SchwerdtnerV2020}. In view of the following theorem, we briefly recall the transfer function of the parameterized system as defined in~\eqref{eq:paramTF}:
$$ \ptf(s, p; \theta) = \pC(p, \theta){\left(sI-\pA(p, \theta)\right)}^{-1}\pB(p, \theta)+\pD(p, \theta),$$
where $\pA(p,\theta) = (\pJ(p,\theta)-\pR(p,\theta))\pQ(p,\theta)$. As both FOM and ROM are assumed to be asymptotically stable, $sI-\fA(p)$ and $sI-\pA(p, \theta)$ are invertible for any~${s\in \overline{\C^+}}$ and~${p\in \Omega}$. In the following theorem, for ease of notation, we assume $\kappaB = \kappaC = \kappaD = \kappaJ = \kappaR = \kappaQ = \kappa$. An extension to the general case with distinct numbers of ansatz functions is straightforward.

\begin{theorem} 
  \label{thm::gradients}
  Let $\theta_0 \in \R^{n_\theta}$, $s_0 \in \overline{\C^+}$, and $p^{(0)} \in \Omega$ be given. 
  Assume that the $j$-th singular value of $\tf(s_0,p^{(0)})-\ptf(s_0,p^{(0)};\theta_0)$ is nonzero and simple and denote by $\vv$ and $\uu$ the corresponding left and right singular vectors. Then the function $${\theta \mapsto \sigma_j(\tf(s_0,p^{(0)})-\ptf(s_0,p^{(0)};\theta))}$$ is differentiable in a neighborhood of $\theta_0$
  and the gradient is given by
  \begin{align*}
    \nabla_\theta \sigma_j(\tf(s_0,p^{(0)})-\ptf(s_0,p^{(0)};\theta_0))=\begin{bmatrix}
      \gradB^\T,\, \gradC^\T,\, \gradD^\T, \gradJ^\T,\gradR^\T,\gradQ^\T
    \end{bmatrix}^\T,
  \end{align*}
  where, according to the partitioning \eqref{eq:partitioning}, 
  \begin{alignat*}{2}
    \gradB^\T &= \begin{bmatrix}
      \dd \theta_{B_0},\ldots,\dd \theta_{B_{\kappaB}}\end{bmatrix}, \qquad \gradC^\T &&= \begin{bmatrix}
    \dd \theta_{C_0},\ldots,\dd \theta_{C_{\kappaC}}\end{bmatrix},\\
      \gradD^\T &= \begin{bmatrix}
        \dd \theta_{D_0},\ldots,\dd \theta_{D_{\kappaD}}\end{bmatrix}, \qquad \gradJ^\T&& = \begin{bmatrix}
      \dd \theta_{J_0},\ldots,\dd \theta_{J_{\kappaJ}}\end{bmatrix},\\
        \gradR^\T &= \begin{bmatrix}
          \dd \theta_{R_0},\ldots,\dd \theta_{R_{\kappaR}}\end{bmatrix}, \qquad \gradQ^\T&& = \begin{bmatrix}
          \dd \theta_{Q_0},\ldots,\dd \theta_{Q_{\kappaQ}}\\
        \end{bmatrix},
        \end{alignat*}
        with
        \begin{subequations}
          \begin{align}
            \gradBj &= -\Real (\ftv({(\vv\uu^\H\pC(p^{(0)}, \theta_0)\dyn^{-1}{f^B_j(p^{(0)})})}^\top)) \label{eq:gradB}  &&\text{ for } j \in \{1, \dots, \kappaB \},\\
            \gradCj &= -\Real (\ftv({(\dyn^{-1}\pB(p^{(0)},\theta_0)\vv\uu^\H{f^C_j(p^{(0)})})}^\top)) \label{eq:gradC}  &&\text{ for } j \in \{1, \dots, \kappaC \},\\
            \gradDj &= -\Real(\ftv({(\vv\uu^\H f^D_j(p^{(0)}))}^\top)) \label{eq:gradD} &&\text{ for } j \in \{1, \dots, \kappaD \},\\
            \gradRj &= \phantom{-}\Real\left( \utv \left(Y_1f_j^R\vtu(\theta_R) +{(Y_1f_j^R)}^\T\vtu(\theta_R) \right)\right)\label{eq:gradR} &&\text{ for } j \in \{1, \dots, \kappaR \},\\
            \gradJj &= \phantom{-}\Real( \sutv(Y_1f_j^J)-\sutv({(Y_1f_j^J)}^\top)) \label{eq:gradJ} &&\text{ for } j \in \{1, \dots, \kappaJ \},\\
            \gradQj&= -\Real( \utv( {(Y_2f_j^Q)}^\top \vtu(\theta_0)+ Y_2f_j^Q\vtu(\theta_0))),\label{eq:gradQ}  &&\text{ for } j \in \{1, \dots, \kappaQ \}
          \end{align}
        \end{subequations}
and 
        \begin{align*}
          \dyn&= s_0I-\pA(p^{(0)},\theta_0),\\
          Y_1 &=\pQ(p^{(0)},\theta_0)\dyn^{-1}  \pB(p^{(0)},\theta_0) \vv \uu^\H \pC(p^{(0)},\theta_0) \dyn^{-1},\\
          Y_2 &=\dyn^{-1}\pB(p^{(0)},\theta_0)\vv\uu^\H\pC(p^{(0)},\theta_0)\dyn^{-1} (\pJ(p^{(0)},\theta) - \pR(p^{(0)},\theta_0)).
        \end{align*}
      \end{theorem}
\begin{proof}
  The proof is a straightforward adaption of~\cite[Proof of Theorem 3.1]{SchwerdtnerV2020} and provided in \Cref{sec:gradcomp}.
\end{proof}

\subsection{Adaptive Sampling}\label{sec:adaptiveSampling}

A zero value of $\loss$ implies that the error is less than $\gamma$ at all sample points. In this way, the distribution of the sample points plays an essential role in the success of our method in minimizing the overall $\hinflinf$ error. If the sample points are chosen poorly, peaks in the error transfer function may be missed entirely and a minimization of $\loss$ may not lead to a small $\hinflinf$ error. However, just choosing an abundance of sample points leads to a high computational burden to construct the ROM, especially in the parametric case, in which a multi-dimensional sample space must be considered. This challenge of choosing the sample points occurs already for nonparametric SOBMOR and was treated in~\cite{Schwerdtner2021}, in which an adaptive sampling procedure that determines frequency sample points (on the imaginary axis) was introduced. For another PMOR sampling strategy based on an error estimate, we refer to~\cite{Binder2021}.

The method in~\cite{Schwerdtner2021} is based on~\cite[Lemma~3]{ApkarianN2018}, which provides a criterion to check if a function $\phi : \R  \rightarrow [0, \infty)$ exceeds a prescribed tolerance from a piecewise linear interpolation of $\phi$ on an interval between two sample points. In~\cite{Schwerdtner2021}, $\phi$ is chosen as the spectral norm of the (nonparametric) error transfer function $\ERR$, i.\,e.,
\begin{align*}
  \ERR: \R  \rightarrow [0, \infty), \quad \omega \mapsto \|\tf(\ri \omega) - \tfr(\ri \omega; \theta)\|_2.
\end{align*}
We propose a method which adds sample points recursively until the deviation between $\ERR$ and its piecewise linear interpolation is sufficiently small for all frequency sample points. This then ensures that the error (over the continuous frequency domain) is captured accurately enough within our set of sample points.

We treat the frequency and parameter domain as a single, $\ell$-dimensional sample space, where $\ell = \dimp +1$ and extend the one-dimensional method in~\cite{Schwerdtner2021} to $\ell$-dimensional grids. In particular, in what follows, we only consider one sample vector ${z = (\omega, p) \in \R \times \pdom}$.

\begin{lemma}[Guarantees along the edges]\label{lem:guaranteesAlongEdges}
  Consider $\phi : \R^{\ell} \to [0,\infty)$ continuously differentiable. Let $h>0$ and consider two neighboring grid points $x,y\in \mathbb{R}^\ell$ with $(x,y) = (x,x+he_i)$, $i=1,\ldots,\ell$, and a first-order upper bound satisfying $L[x,y] \geq \left\vert\frac{\partial}{\partial z_i} \phi(z)\right\vert$ for all $z\in [x,y]$, where $[x,y]$ is the line segment between $x$ and $y$.
  Then if
  \[
    L[x,y]h <2\gamma^* + 2\tau - \phi(y)-\phi(x)
  \]
  with $\gamma^* \geq \max \{\phi(x),\phi(y)\}$, we have
  \[
    \phi(z) < \gamma^* + \tau \qquad \text{for all }z \in [x,y].
  \]
\end{lemma}
\begin{proof}
  By means of the mean value theorem, consider $\xi_1,\xi_2 \in  [x,y]$ such that
  \begin{align}
    \phi(z) - \phi(x) = \nabla \phi(\xi_1)^\top (z-x) \quad \text{and} \quad\phi(z) - \phi(y) = \nabla \phi(\xi_2)^\top (z-y). \label{eq:proof}
  \end{align}
  As $z-x = x - x + she_i = she_i$ and $z-y = x+she_i - x + he_i = (s-1)he_i$  for some $s\in (0,1)$, we have, adding the two equalities in \eqref{eq:proof} that
  \begin{align*}
    2\phi(z) - \phi(x) -\phi(y) &= h\left((s\frac{\partial}{\partial z_i} \phi(\xi_1) + (s-1)\frac{\partial}{\partial z_i} \phi(\xi_2)\right)\\
                                & \leq h\left(s\left\vert \frac{\partial}{\partial z_i} \phi(\xi_1)\right\vert + (1-s) \left\vert \frac{\partial}{\partial z_i} \phi(\xi_2)\right\vert\right).
  \end{align*}
  The derivatives can now be estimated by means of the first-order upper bound and the corresponding assumption, yielding 
  \begin{align*}
    2\phi(z) - \phi(x) -\phi(y) \leq L[x,y]h< 2\gamma^* + 2\tau - \phi(y)-\phi(x).
  \end{align*}
  Canceling the term $ - \phi(x) -\phi(y)$ and dividing by two yields the result.
\end{proof}
The following lemma also provides an error estimate inside the respective cells of the sampling grid. 
\begin{lemma}[Guarantees inside the cells]\label{lem:guaranteesInside}
  Consider $\phi : \R^{\ell} \to [0,\infty)$ continuously differentiable and an $\ell$-dimensional hyperrectangle $H= \bigtimes_{i=1}^{\ell} [a_i,b_i]$ for $a_i,b_i \in \mathbb{R}$ with $a_i<b_i$, $i=1,\ldots,\ell$ with corners $\{z_1,\ldots,z_{2^\ell}\}$ and side lengths $\delta z_i=b_i-a_i$, $i=1,\ldots,\ell$. Further consider $\ell$ first-order upper bounds satisfying $L_i[z_1,\ldots,z_{2^\ell}] \geq |(\nabla \phi(z))_i|$ for all $z\in H$ and $i=1,\ldots,\ell$.
  Then if
  $$
  \sum_{j=1}^{\ell} L_i[z_1,\ldots,z_{2^\ell}]\delta z_j<2\gamma^* + 2\tau - \frac{2\sum_{i=1}^{2^\ell}\phi(z_i)}{2^{\ell}} \qquad \text{for all }z\in H
  $$
  with $\gamma^* \geq \max_{i=1,\ldots,2^\ell} \phi(z_i)$, we have
  $$
  \phi(z) < \gamma^* + \tau \qquad \text{for all }z \in H.
  $$
\end{lemma}
\begin{proof}
  Let $z\in H$ and consider the mean value theorem for all corners, i.e., for $i\in \{1\,\ldots,2^\ell\}$ there are $\xi_i \in [z,z_i]$ such that
  \begin{align}
  \label{eq:proofequation}
    \phi(z)-\phi(z_i) = \nabla \phi(\xi_i)^\top(z-z_i).
  \end{align}
  Summing up \eqref{eq:proofequation} over all corners $i = 1,\ldots,2^\ell$, we obtain
  \begin{align}
    2^\ell \phi(z) - \sum_{i=1}^{2^\ell} \phi(z_i) 
                                                   \leq \sum_{i=1}^{2^\ell} \sum_{j=1}^{\ell} \left\vert \left(\nabla \phi(\xi_i)\right)_j (z_j-(z_i)^j)\right\vert\label{eq:proofcube1}
  \end{align}
  Consider now a fixed corner $i\in \{1,\ldots,2^\ell\}$. Then, as,
  \begin{align*}
    \left\vert (z_j-(z_i)^j)\right\vert = \begin{cases} z_j-(z_i)^j,\quad \text{if } z_j>(z_i)^j, \\
      (z_i)^j-z_j,\quad \text{if } z_j\leq (z_i)^j ,\\
    \end{cases}
  \end{align*}
  we have that $\sum_{i=1}^{2^\ell}\sum_{j=1}^{\ell}\left\vert (z_j-(z_i)^j)\right\vert \leq \frac{2^\ell}{2}\sum_{j=1}^{\ell} \delta z_j$, as $z_i$ are the corners of the hyperrectangle and $z\in H$. Thus,
  \begin{align*}
    \sum_{i=1}^{2^\ell} \sum_{j=1}^{\ell} \left\vert \left(\nabla \phi(\xi_i)\right)_j (z_j-(z_i)^j)\right\vert &\leq  \sum_{i=1}^{2^\ell} \sum_{j=1}^{\ell} \max_{\xi\in H} \left\vert(\nabla(\xi))_i\right\vert \left\vert(z_j-(z_i)^j)\right\vert\\
                                                                                                                &\leq \frac{2^\ell}{2} \sum_{j=1}^{\ell} L_i[z_1,\ldots,z_{2^\ell}]\delta z_j
                                                                                                                < 2^\ell\left(\gamma^* + \tau\right) - \sum_{j=1}^{2^\ell} \phi(z_i)
  \end{align*}
  Invoking \eqref{eq:proofcube1}, cancelling $\sum_{j=1}^{2^\ell} \phi(z_i)$ and dividing by $2^{\ell}$, the claim follows.
\end{proof}

Our adaptive sampling method is described \cref{alg:sampling}. It consists of two steps, that are applied recursively until no new sampling points are added. First, we check along each edge, whether a new sample point must be added based on our criterion from \cref{lem:guaranteesAlongEdges}. For this, we set $\phi$ to the spectral norm of the parametric error transfer function
\begin{align*}
  \ERR_p: \R \times \pdom  \rightarrow [0, \infty), \quad (\omega, p) \mapsto \|\tf(\ri \omega, p) - \tfr(\ri \omega, p; \theta)\|_2.
\end{align*}
For the upper bound of the derivative on an edge we use the maximum of the difference quotients between the center of the edge and its endpoints. After that, if a new point is added, we add edges between the added point and its neighboring points. In this way, we add new points not only on the edges, but also cover the domain using the added edges. In \cref{alg:sampling} the grid $\mathcal{G} = (\mathcal{V}, \mathcal{E})$ consists of a set of vertices $\mathcal{V}\subset\R^\ell$ and edges $\mathcal{E}$. We denote the vertices of edge $e\in \mathcal{E}$ by $(z_1, z_2)$ and the edge connecting the nodes $z_1,z_2\in \mathcal{V}$ by $v(e)$. We define the neighbors of a vertex as all points that differ from the vertex in only one component (i.e., have \textit{Hamming distance} one~\cite[Chapter 2]{Niederreiter2002}) and additionally are not separated by an edge from the given vertex\footnote{an implementation of \cref{alg:sampling} is available at~\url{https://github.com/Algopaul/AdaptiveMesh}}.


\begin{algorithm}[bt]
\caption{Adaptive sampling}\label{alg:sampling}
 \hspace*{\algorithmicindent} \textbf{Input:} Error function $\ERR_p : \mathbb{R}^\ell \to [0,\infty)$, grid $\mathcal{G} = (\mathcal{V},\mathcal{E})$, error level $\gamma > 0$. \\
 \hspace*{\algorithmicindent} \textbf{Output:} Adapted grid $\mathcal{G_+}$
\begin{algorithmic}[1]
	\STATE Set newPointsAdded = true.
	\WHILE{newPointsAdded}
	\STATE Set newPointsAdded = false.
	\FOR{$e \in \mathcal{E}$}
	\STATE Set $(z_1,z_2) = v(e)$.
	\STATE Set $z_\text{test} = z_1 + \tfrac{z_2-z_1}{2}$.
	\STATE Set $d_1 := \left(\ERR_p(z_\text{test}) -\ERR_p(z_1)\right)/(z_\text{test} - z_1)$.
	\STATE Set $d_2 := \left(\ERR_p(z_2 - z_\text{test})\right)/(z_2 - z_\text{test})$.
	\STATE Set $\gamma^* = \max \{\ERR_p(z_1),\ERR_p(z_2)\}$.
	\STATE Set $d^* = \max \{d_1,d_2\}$.
	\IF{$d^*(z_2 - z_1) \geq 2(\gamma + \gamma^*) - \ERR_p(z_1) - \ERR_p(z_2)$}
	\STATE Set $\mathcal{V} = \mathcal{V} \cup \{z_\text{test}\}$.
	\STATE Set newPointsAdded = true.
	\STATE Set $\mathcal{E} = \mathcal{E} \setminus (z_1,z_2)$
	\FOR{$\hat{z} \in \text{neighbor}(z_\text{test})$}
	\IF{$(z_\text{test},\hat{z})\notin \mathcal{E}$}
	\STATE Set $\mathcal{E} = \mathcal{E} \cup (z_\text{test},\hat{z})$.
	\ENDIF
	\ENDFOR
	\ENDIF
	\ENDFOR
	\ENDWHILE
        \STATE Return refined grid $\mathcal{G_+} = (\mathcal{V}, \mathcal{E})$.
\end{algorithmic}
\end{algorithm}

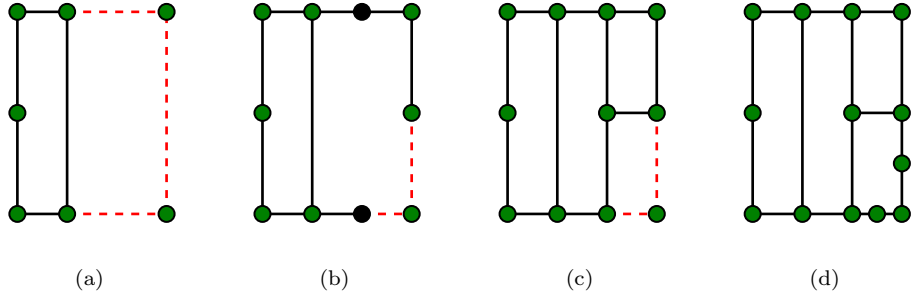
\begin{figure}[ht]
  \centering
    \begin{tabular}{cccc}

\begin{tikzpicture}[/tikz/background rectangle/.style={fill={rgb,1:red,1.0;green,1.0;blue,1.0}, draw opacity={1.0}}, show background rectangle]
\begin{axis}[point meta max={nan}, point meta min={nan}, legend cell align={left}, legend columns={1}, title={}, title style={at={{(0.5,1)}}, anchor={south}, font={{\fontsize{14 pt}{18.2 pt}\selectfont}}, color={rgb,1:red,0.0;green,0.0;blue,0.0}, draw opacity={1.0}, rotate={0.0}}, legend style={color={rgb,1:red,0.0;green,0.0;blue,0.0}, draw opacity={1.0}, line width={1}, solid, fill={rgb,1:red,1.0;green,1.0;blue,1.0}, fill opacity={1.0}, text opacity={1.0}, font={{\fontsize{8 pt}{10.4 pt}\selectfont}}, text={rgb,1:red,0.0;green,0.0;blue,0.0}, cells={anchor={center}}, at={(1.02, 1)}, anchor={north west}}, axis background/.style={fill={rgb,1:red,1.0;green,1.0;blue,1.0}, opacity={1.0}}, anchor={north west}, xshift={1.0mm}, yshift={-1.0mm}, width={0.25\textwidth}, height={0.3\textwidth}, scaled x ticks={false}, xlabel={}, x tick style={draw=none}, x tick label style={color={rgb,1:red,0.0;green,0.0;blue,0.0}, opacity={1.0}, rotate={0}}, xlabel style={at={(ticklabel cs:0.5)}, anchor=near ticklabel, at={{(ticklabel cs:0.5)}}, anchor={near ticklabel}, font={{\fontsize{11 pt}{14.3 pt}\selectfont}}, color={rgb,1:red,0.0;green,0.0;blue,0.0}, draw opacity={1.0}, rotate={0.0}}, xmajorgrids={false}, xmin={0.91}, xmax={4.09}, xtick={}, xticklabels={}, xtick align={inside}, xticklabel style={}, x grid style={}, axis x line*={left}, x axis line style={draw=none}, scaled y ticks={false}, ylabel={}, y tick style={draw=none}, y tick label style={color={rgb,1:red,0.0;green,0.0;blue,0.0}, opacity={1.0}, rotate={0}}, ylabel style={at={(ticklabel cs:0.5)}, anchor=near ticklabel, at={{(ticklabel cs:0.5)}}, anchor={near ticklabel}, font={{\fontsize{11 pt}{14.3 pt}\selectfont}}, color={rgb,1:red,0.0;green,0.0;blue,0.0}, draw opacity={1.0}, rotate={0.0}}, ymajorgrids={false}, ymin={0.98}, ymax={2.03}, ytick={}, yticklabels={}, ytick align={inside}, yticklabel style={font={{\fontsize{8 pt}{10.4 pt}\selectfont}}, color={rgb,1:red,0.0;green,0.0;blue,0.0}, draw opacity={1.0}, rotate={0.0}}, y grid style={}, axis y line*={left}, y axis line style={draw=none}, colorbar={false}]
    \addplot[color={rgb,1:red,0.0;green,0.0;blue,0.0}, name path={3c0df0d4-cfc8-4c32-bf24-f638b67173ce}, draw opacity={1.0}, line width={1}, solid, forget plot]
        table[row sep={\\}]
        {
            \\
            1.0  1.0  \\
            2.0  1.0  \\
        }
        ;
    \addplot[color={rgb,1:red,0.0;green,0.0;blue,0.0}, name path={479173f6-006b-4087-bd57-ac88b09dc9dc}, draw opacity={1.0}, line width={1}, solid, forget plot]
        table[row sep={\\}]
        {
            \\
            1.0  1.0  \\
            1.0  2.0  \\
        }
        ;
    \addplot[color={rgb,1:red,0.0;green,0.0;blue,0.0}, name path={6e3e5afa-1960-4f8c-bf41-4bc1c9817c73}, draw opacity={1.0}, line width={1}, solid, forget plot]
        table[row sep={\\}]
        {
            \\
            1.0  2.0  \\
            2.0  2.0  \\
        }
        ;
    \addplot[color={rgb,1:red,0.0;green,0.0;blue,0.0}, name path={de0bcf34-8cca-4c82-912d-b41d71503548}, draw opacity={1.0}, line width={1}, solid, forget plot]
        table[row sep={\\}]
        {
            \\
            2.0  1.0  \\
            2.0  2.0  \\
        }
        ;
    \addplot[color={rgb,1:red,1.0;green,0.0;blue,0.0}, name path={4905f7ee-88e0-40c0-90a7-54b17797acac}, draw opacity={1.0}, line width={1}, dashed, forget plot]
        table[row sep={\\}]
        {
            \\
            2.0  2.0  \\
            4.0  2.0  \\
        }
        ;
    \addplot[color={rgb,1:red,1.0;green,0.0;blue,0.0}, name path={4905f7ee-88e0-40c0-90a7-54b17797acac}, draw opacity={1.0}, line width={1}, dashed, forget plot]
        table[row sep={\\}]
        {
            \\
            2.0  1.0  \\
            4.0  1.0  \\
        }
        ;
    \addplot[color={rgb,1:red,1.0;green,0.0;blue,0.0}, name path={ec6f74d0-2b55-418e-9a2f-a79205f4e6ec}, draw opacity={1.0}, line width={1}, dashed, forget plot]
        table[row sep={\\}]
        {
            \\
            4.0  1.0  \\
            4.0  2.0  \\
        }
        ;
    \addplot[color={rgb,1:red,0.4231;green,0.6225;blue,0.1988}, name path={7a4bbad6-977f-44a9-b65f-39687012977b}, only marks, draw opacity={1.0}, line width={0}, solid, mark={*}, mark size={3.0 pt}, mark repeat={1}, mark options={color={rgb,1:red,0.0;green,0.0;blue,0.0}, draw opacity={1.0}, fill={rgb,1:red,0.0;green,0.502;blue,0.0}, fill opacity={1.0}, line width={0.75}, rotate={0}, solid}, forget plot]
        table[row sep={\\}]
        {
            \\
            1.0  1.0  \\
            1.0  1.5  \\
            1.0  2.0  \\
            2.0  1.0  \\
            2.0  2.0  \\
            4.0  1.0  \\
            4.0  2.0  \\
        }
        ;
\end{axis}
\end{tikzpicture} &

\begin{tikzpicture}[/tikz/background rectangle/.style={fill={rgb,1:red,1.0;green,1.0;blue,1.0}, draw opacity={1.0}}, show background rectangle]
\begin{axis}[point meta max={nan}, point meta min={nan}, legend cell align={left}, legend columns={1}, title={}, title style={at={{(0.5,1)}}, anchor={south}, font={{\fontsize{14 pt}{18.2 pt}\selectfont}}, color={rgb,1:red,0.0;green,0.0;blue,0.0}, draw opacity={1.0}, rotate={0.0}}, legend style={color={rgb,1:red,0.0;green,0.0;blue,0.0}, draw opacity={1.0}, line width={1}, solid, fill={rgb,1:red,1.0;green,1.0;blue,1.0}, fill opacity={1.0}, text opacity={1.0}, font={{\fontsize{8 pt}{10.4 pt}\selectfont}}, text={rgb,1:red,0.0;green,0.0;blue,0.0}, cells={anchor={center}}, at={(1.02, 1)}, anchor={north west}}, axis background/.style={fill={rgb,1:red,1.0;green,1.0;blue,1.0}, opacity={1.0}}, anchor={north west}, xshift={1.0mm}, yshift={-1.0mm}, width={0.25\textwidth}, height={0.3\textwidth}, scaled x ticks={false}, xlabel={}, x tick style={draw=none}, x tick label style={color={rgb,1:red,0.0;green,0.0;blue,0.0}, opacity={1.0}, rotate={0}}, xlabel style={at={(ticklabel cs:0.5)}, anchor=near ticklabel, at={{(ticklabel cs:0.5)}}, anchor={near ticklabel}, font={{\fontsize{11 pt}{14.3 pt}\selectfont}}, color={rgb,1:red,0.0;green,0.0;blue,0.0}, draw opacity={1.0}, rotate={0.0}}, xmajorgrids={false}, xmin={0.91}, xmax={4.09}, xtick={}, xticklabels={}, xtick align={inside}, xticklabel style={}, x grid style={}, axis x line*={left}, x axis line style={draw=none}, scaled y ticks={false}, ylabel={}, y tick style={draw=none}, y tick label style={color={rgb,1:red,0.0;green,0.0;blue,0.0}, opacity={1.0}, rotate={0}}, ylabel style={at={(ticklabel cs:0.5)}, anchor=near ticklabel, at={{(ticklabel cs:0.5)}}, anchor={near ticklabel}, font={{\fontsize{11 pt}{14.3 pt}\selectfont}}, color={rgb,1:red,0.0;green,0.0;blue,0.0}, draw opacity={1.0}, rotate={0.0}}, ymajorgrids={false}, ymin={0.98}, ymax={2.03}, ytick={}, yticklabels={}, ytick align={inside}, yticklabel style={font={{\fontsize{8 pt}{10.4 pt}\selectfont}}, color={rgb,1:red,0.0;green,0.0;blue,0.0}, draw opacity={1.0}, rotate={0.0}}, y grid style={}, axis y line*={left}, y axis line style={draw=none}, colorbar={false}]
    \addplot[color={rgb,1:red,0.0;green,0.0;blue,0.0}, name path={3c0df0d4-cfc8-4c32-bf24-f638b67173ce}, draw opacity={1.0}, line width={1}, solid, forget plot]
        table[row sep={\\}]
        {
            \\
            1.0  1.0  \\
            2.0  1.0  \\
        }
        ;
    \addplot[color={rgb,1:red,0.0;green,0.0;blue,0.0}, name path={479173f6-006b-4087-bd57-ac88b09dc9dc}, draw opacity={1.0}, line width={1}, solid, forget plot]
        table[row sep={\\}]
        {
            \\
            1.0  1.0  \\
            1.0  2.0  \\
        }
        ;
    \addplot[color={rgb,1:red,0.0;green,0.0;blue,0.0}, name path={6e3e5afa-1960-4f8c-bf41-4bc1c9817c73}, draw opacity={1.0}, line width={1}, solid, forget plot]
        table[row sep={\\}]
        {
            \\
            1.0  2.0  \\
            2.0  2.0  \\
        }
        ;
    \addplot[color={rgb,1:red,0.0;green,0.0;blue,0.0}, name path={de0bcf34-8cca-4c82-912d-b41d71503548}, draw opacity={1.0}, line width={1}, solid, forget plot]
        table[row sep={\\}]
        {
            \\
            2.0  1.0  \\
            2.0  2.0  \\
        }
        ;
    \addplot[color={rgb,1:red,0.0;green,0.0;blue,0.0}, name path={4905f7ee-88e0-40c0-90a7-54b17797acac}, draw opacity={1.0}, line width={1}, solid, forget plot]
        table[row sep={\\}]
        {
            \\
            2.0  2.0  \\
            4.0  2.0  \\
        }
        ;
    \addplot[color={rgb,1:red,0.0;green,0.0;blue,0.0}, name path={4905f7ee-88e0-40c0-90a7-54b17797acac}, draw opacity={1.0}, line width={1}, solid, forget plot]
        table[row sep={\\}]
        {
            \\
            2.0  1.0  \\
            3.0  1.0  \\
        }
        ;
    \addplot[color={rgb,1:red,1.0;green,0.0;blue,0.0}, name path={4905f7ee-88e0-40c0-90a7-54b17797acac}, draw opacity={1.0}, line width={1}, dashed, forget plot]
        table[row sep={\\}]
        {
            \\
            3.0  1.0  \\
            4.0  1.0  \\
        }
        ;
    \addplot[color={rgb,1:red,1.0;green,0.0;blue,0.0}, name path={ec6f74d0-2b55-418e-9a2f-a79205f4e6ec}, draw opacity={1.0}, line width={1}, dashed, forget plot]
        table[row sep={\\}]
        {
            \\
            4.0  1.0  \\
            4.0  1.5  \\
        }
        ;
    \addplot[color={rgb,1:red,0.0;green,0.0;blue,0.0}, name path={ec6f74d0-2b55-418e-9a2f-a79205f4e6ec}, draw opacity={1.0}, line width={1}, solid, forget plot]
        table[row sep={\\}]
        {
            \\
            4.0  1.5  \\
            4.0  2.0  \\
        }
        ;
    \addplot[color={rgb,1:red,0.4231;green,0.6225;blue,0.1988}, name path={7a4bbad6-977f-44a9-b65f-39687012977b}, only marks, draw opacity={1.0}, line width={0}, solid, mark={*}, mark size={3.0 pt}, mark repeat={1}, mark options={color={rgb,1:red,0.0;green,0.0;blue,0.0}, draw opacity={1.0}, fill={rgb,1:red,0.0;green,0.502;blue,0.0}, fill opacity={1.0}, line width={0.75}, rotate={0}, solid}, forget plot]
        table[row sep={\\}]
        {
            \\
            1.0  1.0  \\
            1.0  1.5  \\
            1.0  2.0  \\
            2.0  1.0  \\
            2.0  2.0  \\
            4.0  1.0  \\
            4.0  2.0  \\
            4.0  1.5  \\
        }
        ;
    \addplot[color={rgb,1:red,0.0;green,0.0;blue,0.0}, name path={7a4bbad6-977f-44a9-b65f-39687012977b}, only marks, draw opacity={1.0}, line width={0}, solid, mark={*}, mark size={3.0 pt}, mark repeat={1}, mark options={color={rgb,1:red,0.0;green,0.0;blue,0.0}, draw opacity={1.0}, fill={rgb,1:red,0.0;green,0.0;blue,0.0}, fill opacity={1.0}, line width={0.75}, rotate={0}, solid}, forget plot]
        table[row sep={\\}]
        {
            \\
            3.0  1.0  \\
            3.0  2.0  \\
        }
        ;
\end{axis}
\end{tikzpicture} &

\begin{tikzpicture}[/tikz/background rectangle/.style={fill={rgb,1:red,1.0;green,1.0;blue,1.0}, draw opacity={1.0}}, show background rectangle]
\begin{axis}[point meta max={nan}, point meta min={nan}, legend cell align={left}, legend columns={1}, title={}, title style={at={{(0.5,1)}}, anchor={south}, font={{\fontsize{14 pt}{18.2 pt}\selectfont}}, color={rgb,1:red,0.0;green,0.0;blue,0.0}, draw opacity={1.0}, rotate={0.0}}, legend style={color={rgb,1:red,0.0;green,0.0;blue,0.0}, draw opacity={1.0}, line width={1}, solid, fill={rgb,1:red,1.0;green,1.0;blue,1.0}, fill opacity={1.0}, text opacity={1.0}, font={{\fontsize{8 pt}{10.4 pt}\selectfont}}, text={rgb,1:red,0.0;green,0.0;blue,0.0}, cells={anchor={center}}, at={(1.02, 1)}, anchor={north west}}, axis background/.style={fill={rgb,1:red,1.0;green,1.0;blue,1.0}, opacity={1.0}}, anchor={north west}, xshift={1.0mm}, yshift={-1.0mm}, width={0.25\textwidth}, height={0.3\textwidth}, scaled x ticks={false}, xlabel={}, x tick style={draw=none}, x tick label style={color={rgb,1:red,0.0;green,0.0;blue,0.0}, opacity={1.0}, rotate={0}}, xlabel style={at={(ticklabel cs:0.5)}, anchor=near ticklabel, at={{(ticklabel cs:0.5)}}, anchor={near ticklabel}, font={{\fontsize{11 pt}{14.3 pt}\selectfont}}, color={rgb,1:red,0.0;green,0.0;blue,0.0}, draw opacity={1.0}, rotate={0.0}}, xmajorgrids={false}, xmin={0.91}, xmax={4.09}, xtick={}, xticklabels={}, xtick align={inside}, xticklabel style={}, x grid style={}, axis x line*={left}, x axis line style={draw=none}, scaled y ticks={false}, ylabel={}, y tick style={draw=none}, y tick label style={color={rgb,1:red,0.0;green,0.0;blue,0.0}, opacity={1.0}, rotate={0}}, ylabel style={at={(ticklabel cs:0.5)}, anchor=near ticklabel, at={{(ticklabel cs:0.5)}}, anchor={near ticklabel}, font={{\fontsize{11 pt}{14.3 pt}\selectfont}}, color={rgb,1:red,0.0;green,0.0;blue,0.0}, draw opacity={1.0}, rotate={0.0}}, ymajorgrids={false}, ymin={0.98}, ymax={2.03}, ytick={}, yticklabels={}, ytick align={inside}, yticklabel style={font={{\fontsize{8 pt}{10.4 pt}\selectfont}}, color={rgb,1:red,0.0;green,0.0;blue,0.0}, draw opacity={1.0}, rotate={0.0}}, y grid style={}, axis y line*={left}, y axis line style={draw=none}, colorbar={false}]
    \addplot[color={rgb,1:red,0.0;green,0.0;blue,0.0}, name path={3c0df0d4-cfc8-4c32-bf24-f638b67173ce}, draw opacity={1.0}, line width={1}, solid, forget plot]
        table[row sep={\\}]
        {
            \\
            1.0  1.0  \\
            2.0  1.0  \\
        }
        ;
    \addplot[color={rgb,1:red,0.0;green,0.0;blue,0.0}, name path={479173f6-006b-4087-bd57-ac88b09dc9dc}, draw opacity={1.0}, line width={1}, solid, forget plot]
        table[row sep={\\}]
        {
            \\
            1.0  1.0  \\
            1.0  2.0  \\
        }
        ;
    \addplot[color={rgb,1:red,0.0;green,0.0;blue,0.0}, name path={6e3e5afa-1960-4f8c-bf41-4bc1c9817c73}, draw opacity={1.0}, line width={1}, solid, forget plot]
        table[row sep={\\}]
        {
            \\
            1.0  2.0  \\
            2.0  2.0  \\
        }
        ;
    \addplot[color={rgb,1:red,0.0;green,0.0;blue,0.0}, name path={de0bcf34-8cca-4c82-912d-b41d71503548}, draw opacity={1.0}, line width={1}, solid, forget plot]
        table[row sep={\\}]
        {
            \\
            2.0  1.0  \\
            2.0  2.0  \\
        }
        ;
    \addplot[color={rgb,1:red,0.0;green,0.0;blue,0.0}, name path={4905f7ee-88e0-40c0-90a7-54b17797acac}, draw opacity={1.0}, line width={1}, solid, forget plot]
        table[row sep={\\}]
        {
            \\
            2.0  2.0  \\
            4.0  2.0  \\
        }
        ;
    \addplot[color={rgb,1:red,0.0;green,0.0;blue,0.0}, name path={4905f7ee-88e0-40c0-90a7-54b17797acac}, draw opacity={1.0}, line width={1}, solid, forget plot]
        table[row sep={\\}]
        {
            \\
            2.0  1.0  \\
            3.0  1.0  \\
        }
        ;
    \addplot[color={rgb,1:red,1.0;green,0.0;blue,0.0}, name path={4905f7ee-88e0-40c0-90a7-54b17797acac}, draw opacity={1.0}, line width={1}, dashed, forget plot]
        table[row sep={\\}]
        {
            \\
            3.0  1.0  \\
            4.0  1.0  \\
        }
        ;
    \addplot[color={rgb,1:red,1.0;green,0.0;blue,0.0}, name path={ec6f74d0-2b55-418e-9a2f-a79205f4e6ec}, draw opacity={1.0}, line width={1}, dashed, forget plot]
        table[row sep={\\}]
        {
            \\
            4.0  1.0  \\
            4.0  1.5  \\
        }
        ;
    \addplot[color={rgb,1:red,0.0;green,0.0;blue,0.0}, name path={ec6f74d0-2b55-418e-9a2f-a79205f4e6ec}, draw opacity={1.0}, line width={1}, solid, forget plot]
        table[row sep={\\}]
        {
            \\
            4.0  1.5  \\
            4.0  2.0  \\
        }
        ;
    \addplot[color={rgb,1:red,0.0;green,0.0;blue,0.0}, name path={ec6f74d0-2b55-418e-9a2f-a79205f4e6ec}, draw opacity={1.0}, line width={1}, solid, forget plot]
        table[row sep={\\}]
        {
            \\
            3.0  1.0  \\
            3.0  2.0  \\
        }
        ;
    \addplot[color={rgb,1:red,0.0;green,0.0;blue,0.0}, name path={ec6f74d0-2b55-418e-9a2f-a79205f4e6ec}, draw opacity={1.0}, line width={1}, solid, forget plot]
        table[row sep={\\}]
        {
            \\
            3.0  1.5  \\
            4.0  1.5  \\
        }
        ;
    \addplot[color={rgb,1:red,0.4231;green,0.6225;blue,0.1988}, name path={7a4bbad6-977f-44a9-b65f-39687012977b}, only marks, draw opacity={1.0}, line width={0}, solid, mark={*}, mark size={3.0 pt}, mark repeat={1}, mark options={color={rgb,1:red,0.0;green,0.0;blue,0.0}, draw opacity={1.0}, fill={rgb,1:red,0.0;green,0.502;blue,0.0}, fill opacity={1.0}, line width={0.75}, rotate={0}, solid}, forget plot]
        table[row sep={\\}]
        {
            \\
            1.0  1.0  \\
            1.0  1.5  \\
            1.0  2.0  \\
            2.0  1.0  \\
            2.0  2.0  \\
            4.0  1.0  \\
            4.0  2.0  \\
            3.0  1.0  \\
            3.0  2.0  \\
            4.0  1.5  \\
            3.0  1.5  \\
        }
        ;
\end{axis}
\end{tikzpicture} &

\begin{tikzpicture}[/tikz/background rectangle/.style={fill={rgb,1:red,1.0;green,1.0;blue,1.0}, draw opacity={1.0}}, show background rectangle]
\begin{axis}[point meta max={nan}, point meta min={nan}, legend cell align={left}, legend columns={1}, title={}, title style={at={{(0.5,1)}}, anchor={south}, font={{\fontsize{14 pt}{18.2 pt}\selectfont}}, color={rgb,1:red,0.0;green,0.0;blue,0.0}, draw opacity={1.0}, rotate={0.0}}, legend style={color={rgb,1:red,0.0;green,0.0;blue,0.0}, draw opacity={1.0}, line width={1}, solid, fill={rgb,1:red,1.0;green,1.0;blue,1.0}, fill opacity={1.0}, text opacity={1.0}, font={{\fontsize{8 pt}{10.4 pt}\selectfont}}, text={rgb,1:red,0.0;green,0.0;blue,0.0}, cells={anchor={center}}, at={(1.02, 1)}, anchor={north west}}, axis background/.style={fill={rgb,1:red,1.0;green,1.0;blue,1.0}, opacity={1.0}}, anchor={north west}, xshift={1.0mm}, yshift={-1.0mm}, width={0.25\textwidth}, height={0.3\textwidth}, scaled x ticks={false}, xlabel={}, x tick style={draw=none}, x tick label style={color={rgb,1:red,0.0;green,0.0;blue,0.0}, opacity={1.0}, rotate={0}}, xlabel style={at={(ticklabel cs:0.5)}, anchor=near ticklabel, at={{(ticklabel cs:0.5)}}, anchor={near ticklabel}, font={{\fontsize{11 pt}{14.3 pt}\selectfont}}, color={rgb,1:red,0.0;green,0.0;blue,0.0}, draw opacity={1.0}, rotate={0.0}}, xmajorgrids={false}, xmin={0.91}, xmax={4.09}, xtick={}, xticklabels={}, xtick align={inside}, xticklabel style={}, x grid style={}, axis x line*={left}, x axis line style={draw=none}, scaled y ticks={false}, ylabel={}, y tick style={draw=none}, y tick label style={color={rgb,1:red,0.0;green,0.0;blue,0.0}, opacity={1.0}, rotate={0}}, ylabel style={at={(ticklabel cs:0.5)}, anchor=near ticklabel, at={{(ticklabel cs:0.5)}}, anchor={near ticklabel}, font={{\fontsize{11 pt}{14.3 pt}\selectfont}}, color={rgb,1:red,0.0;green,0.0;blue,0.0}, draw opacity={1.0}, rotate={0.0}}, ymajorgrids={false}, ymin={0.98}, ymax={2.03}, ytick={}, yticklabels={}, ytick align={inside}, yticklabel style={font={{\fontsize{8 pt}{10.4 pt}\selectfont}}, color={rgb,1:red,0.0;green,0.0;blue,0.0}, draw opacity={1.0}, rotate={0.0}}, y grid style={}, axis y line*={left}, y axis line style={draw=none}, colorbar={false}]
    \addplot[color={rgb,1:red,0.0;green,0.0;blue,0.0}, name path={3c0df0d4-cfc8-4c32-bf24-f638b67173ce}, draw opacity={1.0}, line width={1}, solid, forget plot]
        table[row sep={\\}]
        {
            \\
            1.0  1.0  \\
            2.0  1.0  \\
        }
        ;
    \addplot[color={rgb,1:red,0.0;green,0.0;blue,0.0}, name path={479173f6-006b-4087-bd57-ac88b09dc9dc}, draw opacity={1.0}, line width={1}, solid, forget plot]
        table[row sep={\\}]
        {
            \\
            1.0  1.0  \\
            1.0  2.0  \\
        }
        ;
    \addplot[color={rgb,1:red,0.0;green,0.0;blue,0.0}, name path={6e3e5afa-1960-4f8c-bf41-4bc1c9817c73}, draw opacity={1.0}, line width={1}, solid, forget plot]
        table[row sep={\\}]
        {
            \\
            1.0  2.0  \\
            2.0  2.0  \\
        }
        ;
    \addplot[color={rgb,1:red,0.0;green,0.0;blue,0.0}, name path={de0bcf34-8cca-4c82-912d-b41d71503548}, draw opacity={1.0}, line width={1}, solid, forget plot]
        table[row sep={\\}]
        {
            \\
            2.0  1.0  \\
            2.0  2.0  \\
        }
        ;
    \addplot[color={rgb,1:red,0.0;green,0.0;blue,0.0}, name path={4905f7ee-88e0-40c0-90a7-54b17797acac}, draw opacity={1.0}, line width={1}, solid, forget plot]
        table[row sep={\\}]
        {
            \\
            2.0  2.0  \\
            4.0  2.0  \\
        }
        ;
    \addplot[color={rgb,1:red,0.0;green,0.0;blue,0.0}, name path={4905f7ee-88e0-40c0-90a7-54b17797acac}, draw opacity={1.0}, line width={1}, solid, forget plot]
        table[row sep={\\}]
        {
            \\
            2.0  1.0  \\
            3.0  1.0  \\
        }
        ;
    \addplot[color={rgb,1:red,0.0;green,0.0;blue,0.0}, name path={4905f7ee-88e0-40c0-90a7-54b17797acac}, draw opacity={1.0}, line width={1}, solid, forget plot]
        table[row sep={\\}]
        {
            \\
            3.0  1.0  \\
            4.0  1.0  \\
        }
        ;
    \addplot[color={rgb,1:red,0.0;green,0.0;blue,0.0}, name path={ec6f74d0-2b55-418e-9a2f-a79205f4e6ec}, draw opacity={1.0}, line width={1}, solid, forget plot]
        table[row sep={\\}]
        {
            \\
            4.0  1.0  \\
            4.0  1.5  \\
        }
        ;
    \addplot[color={rgb,1:red,0.0;green,0.0;blue,0.0}, name path={ec6f74d0-2b55-418e-9a2f-a79205f4e6ec}, draw opacity={1.0}, line width={1}, solid, forget plot]
        table[row sep={\\}]
        {
            \\
            4.0  1.5  \\
            4.0  2.0  \\
        }
        ;
    \addplot[color={rgb,1:red,0.0;green,0.0;blue,0.0}, name path={ec6f74d0-2b55-418e-9a2f-a79205f4e6ec}, draw opacity={1.0}, line width={1}, solid, forget plot]
        table[row sep={\\}]
        {
            \\
            3.0  1.0  \\
            3.0  2.0  \\
        }
        ;
    \addplot[color={rgb,1:red,0.0;green,0.0;blue,0.0}, name path={ec6f74d0-2b55-418e-9a2f-a79205f4e6ec}, draw opacity={1.0}, line width={1}, solid, forget plot]
        table[row sep={\\}]
        {
            \\
            3.0  1.5  \\
            4.0  1.5  \\
        }
        ;
    \addplot[color={rgb,1:red,0.4231;green,0.6225;blue,0.1988}, name path={7a4bbad6-977f-44a9-b65f-39687012977b}, only marks, draw opacity={1.0}, line width={0}, solid, mark={*}, mark size={3.0 pt}, mark repeat={1}, mark options={color={rgb,1:red,0.0;green,0.0;blue,0.0}, draw opacity={1.0}, fill={rgb,1:red,0.0;green,0.502;blue,0.0}, fill opacity={1.0}, line width={0.75}, rotate={0}, solid}, forget plot]
        table[row sep={\\}]
        {
            \\
            1.0  1.0  \\
            1.0  1.5  \\
            1.0  2.0  \\
            2.0  1.0  \\
            2.0  2.0  \\
            4.0  1.0  \\
            4.0  2.0  \\
            3.0  1.0  \\
            3.0  2.0  \\
            4.0  1.5  \\
            3.0  1.5  \\
            3.5  1.0  \\
            4.0  1.25  \\
        }
        ;
\end{axis}
\end{tikzpicture} \\
      \scriptsize{(a)} & \scriptsize{(b)} & \scriptsize{(c)} & \scriptsize{(d)} \\
    \end{tabular}
  \caption{Example of the adaptive sampling algorithm for a 2D grid.}\label{fig:adaptillustration}
\end{figure}

In \cref{fig:adaptillustration} (a-d), we illustrate the adaptive sampling algorithm for a 2D grid. In~(a) we show the grid before the adaptive sampling algorithm is applied. Edges with small enough error, i.e., that do not satisfy the condition in line 11 of \cref{alg:sampling} and thus need no further division are shown as black solid line. Edges that need to be divided according to line 11 of \cref{alg:sampling} are illustrated as red dashed line. In~(b), new sample points are added at the center of the red edges. After that, we check if further sample points must be added because there are two neighboring points that are not connected. This is true for the two black points in~(b). Therefore a point is added and connected to all its neighboring points in~(c). Since new points are added, the loop in lines 2--19 of \cref{alg:sampling} is executed again and two more points are added at the remaining red dashed edges. After that, the adaptive sampling algorithm terminates.

\section{Numerical Experiments}\label{sec:numerics}
In the following numerical experiments we showcase the effectiveness of our method. The PMOR methods in this comparison all have a number of hyper-parameters that can be tweaked and which may lead to different results. As an example, in parametric IRKA (pIRKA), the position of the sample points $p^{(i)}$ and dimensions of the corresponding projection subspaces influence the final accuracy. To provide a fair comparison between the different methods, we compare our method to the global basis method of~\cite{BaurBBG2011} in \Cref{sec:Thermal} and the matrix interpolation methods of~\cite{GeussPL2013} in \Cref{sec:Timoshenko} by means of the benchmark systems presented in the respective articles using the same experimental setup. In \Cref{sec:PHMSD}, a third experiment is conducted to demonstrate our method for structure-preserving PMOR\@, to provide a comparison to the optimization-based $\htwo \otimes \ltwo$ reduction of~\cite{HundMMS2021}, and to explore the increase in accuracy for increased reduced model orders. We use the implementation of the \textsf{MATLAB} toolbox \textsf{psssMOR}\footnote{available at \url{https://www.mathworks.com/matlabcentral/fileexchange/64392-psssmor-toolbox}} for the local and global methods described in \Cref{subsec:existing}.



The scalar ansatz functions in our ROMs offer much space for hyper-parameter tuning. In an engineering application, the number of coefficients per matrix function and the type of function can be tuned to achieve the best possible fit. For comparibility in our experiments, we only use rather general purpose linear hat functions as ansatz functions. As short-hand notation, for $a, b \in \R$ with $a<b$, we define
\begin{align}\label{eq:hat}
\hat{f}(x,a,b) = \begin{cases}
\phantom{-}2(x-a)/(b-a) & \mathrm{if}\, x \in  [a,a+(b-a)/2]\\
-2(x-b)/(b-a) & \mathrm{if}\, x \in  [a+(b-a)/2,b]\\
\phantom{-}0&\text{otherwise}
\end{cases}
\end{align}
 This continuous function $\hat{f}(x, a, b)$ is zero for $x\leq a$ and $x\geq b$ and describes a piecewise linear function between $a$ and $b$ with maximum value one for $(a+b)/2$.



\subsection{Thermal Model}\label{sec:Thermal}
In our first experiment, we compare our suggested meth\-od to the global basis approach suggested in \cite{BaurBBG2011}.
We use the same model as in in~\cite[Section~6.2]{BaurBBG2011}, which describes the thermal conduction in a semiconductor chip that is connected to its environment with three device interfaces. The modeling process is explained in~\cite{Lasance2001}. The model equations are given by
\begin{align}
  \thermM :
  \begin{cases}
  \thermE \dot x(t) = \left(\thermA_0 - \sum\limits_{i=1}^3 p_i \thermA_i \right)x(t)+\thermB u(t),\\ 
  \phantom{\thermE} y(t) = \thermC x(t),
  \end{cases}
\end{align}
where $\thermE, \thermA_i \in \R^{4257 \times 4257}$, $i=1,\ldots,3$, $\thermB \in \R^{4257}$, and $\thermC \in \R^{7 \times 4257}$, which are available at the \textsf{MORWiki}\footnote{\url{https://morwiki.mpi-magdeburg.mpg.de/morwiki/index.php/Thermal_Model}}. 
The parameters $p_i, i=1\ldots,3$, are scalar and contained in the parameter range $\Omega = [1, 10^4]$. This model is used in~\cite{BaurBBG2011} to showcase the effectiveness of the PMOR method described therein. 

In~\cite{BaurBBG2011}, a parametric ROM of order 14 is obtained from reducing $\thermM$ for fixed parameter values $p^{(1)} = \left[ 10^4, 10^4, 1 \right]^\T$, $p^{(2)} = \left[ 1, 1, 1\right]^\T$. Here, we use the \textsf{psssMOR} implementation of pIRKA to compute a projection-based parametric ROM with the same dimension and parameter samples. 
We obtain slightly lower $\hinf$ errors than reported in~\cite[Fig.~6.12]{BaurBBG2011}, even though we use the same algorithm and setup. Note that we report the absolute $\hinf$ errors in \cref{fig:ThermalSurfacePlot}, while in~\cite{BaurBBG2011} the relative $\hinf$ errors are shown. As in~\cite{BaurBBG2011} we regard the third parameter $p_3=1$ as fixed and do not vary it during the reduction or evaluation.

In our method, we use two different parameterizations. First, we aim at a parametric ROM with the same complexity as the model obtained with pIRKA\@. 
For this, we use constant input, output, and feedthrough matrices. After that, we allow for a parameter dependence in all matrices. The parameterizations are described in \Cref{romdef:thermala} and \Cref{romdef:thermalb}, respectively.

\begin{romdef}\label{romdef:thermala}
  We use the ROM setup defined in~\eqref{eq:romform} with the matrix-valued functions as described in~\eqref{eq:ansatz}. The ROM dimensions are set to $r=14$, $n_y=7$, and $n_u=1$. We set ${\kappa_B=\kappa_C=\kappa_D=1}$ and use the ansatz functions ${f_1^B \equiv f_1^C \equiv f_1^D \equiv 1}$ for the input, output, and feedthrough matrices. For $\pJ, \pR$, and $\pQ$, we define the functions
  \begin{align}
    \label{eq:thermfs}
    f_1 = \hat{f}( \cdot, 2.0-10^4, 10^4), \quad f_2 = \hat{f}( \cdot, 1.0, 2  \cdot 10^4 -1),
  \end{align}
  and set ${\kappa_J=\kappa_R=\kappa_Q=4}$ and use the ansatz functions
  \begin{align}
    \label{eq:therma_JRQ}
    \begin{split}
      f_1^J(p) \equiv f_1^R(p) \equiv f_1^Q(p) \equiv f_1(p_1), \quad
      f_2^J(p) \equiv f_2^R(p) \equiv f_2^Q(p) \equiv f_1(p_2), \\
      f_3^J(p) \equiv f_3^R(p) \equiv f_3^Q(p) \equiv f_2(p_1), \quad
      f_4^J(p) \equiv f_4^R(p) \equiv f_4^Q(p) \equiv f_2(p_2).
    \end{split}
  \end{align}
\end{romdef}

\begin{romdef}\label{romdef:thermalb}
  We use the ROM setup defined in~\eqref{eq:romform} with the matrix-valued functions as described in~\eqref{eq:ansatz}. The ROM dimensions are set to $r=14$, $p=7$, and $m=1$. We set all ${\kappa_J = \dots = \kappa_D = 4}$, use the same ansatz functions for $f_i^J, f_i^R, f_i^Q$ for $i\in \{1,\dots,4\}$ as in~\eqref{eq:therma_JRQ}, and use additionally
  \begin{align*}
    \begin{split}
      f_1^B(p) \equiv f_1^C(p) \equiv f_1^D(p) \equiv f_1(p_1), \quad
      f_2^B(p) \equiv f_2^C(p) \equiv f_2^D(p) \equiv f_1(p_2), \\
      f_3^B(p) \equiv f_3^C(p) \equiv f_3^D(p) \equiv f_2(p_1), \quad
      f_4^B(p) \equiv f_4^C(p) \equiv f_4^D(p) \equiv f_2(p_2),
    \end{split}
  \end{align*}
  where $f_1$ and $f_2$ are defined as in~\eqref{eq:thermfs}.
\end{romdef}

\begin{figure}[t]
  \centering
  \includegraphics[width=0.9\textwidth]{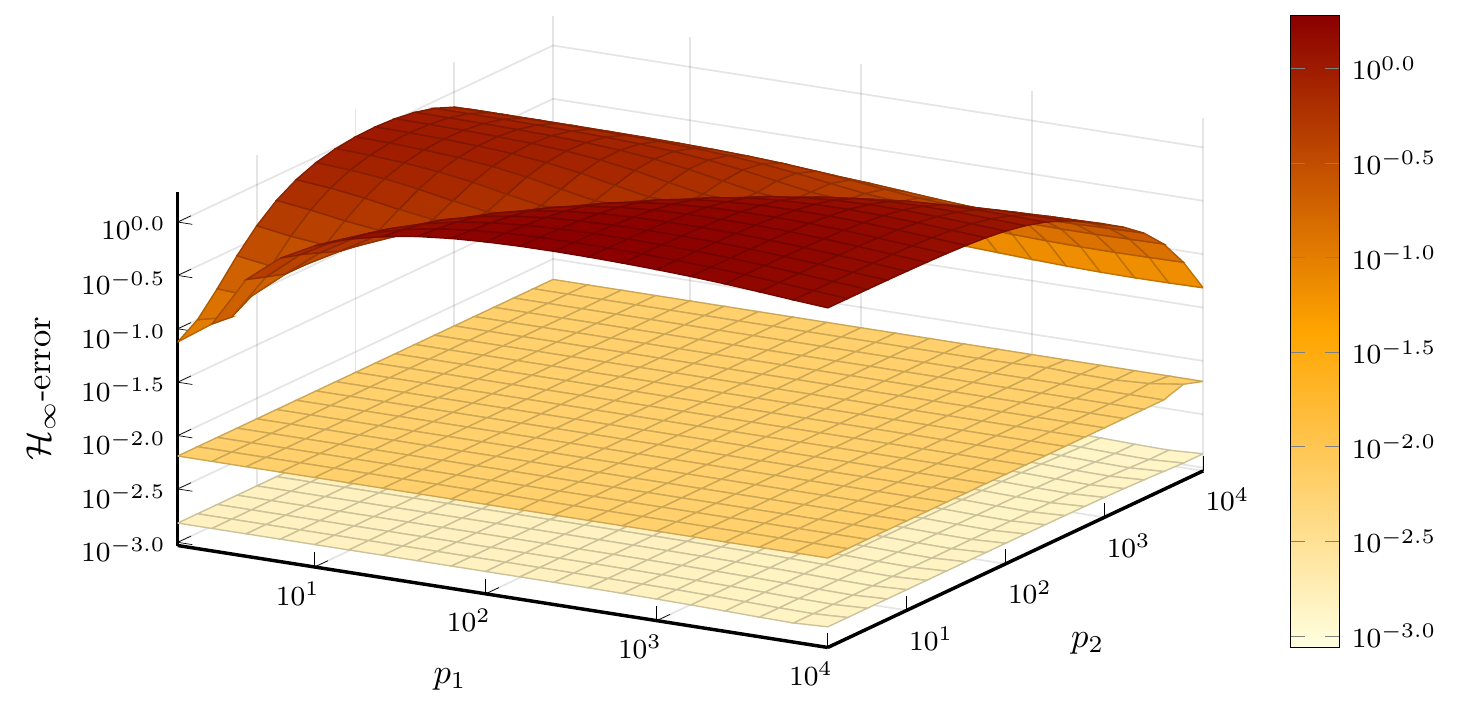}
  \caption{$\mathcal{H}_\infty $ error comparison for varying $p_1$ and $p_2$. The upper surface corresponds to the errors of pIRKA~\cite{BaurBBG2011}, the surface in the middle to the errors of our method with \Cref{romdef:thermala}, and the lower surface to the errors of our method with \Cref{romdef:thermalb}.}\label{fig:ThermalSurfacePlot}
\end{figure}

The comparison between our method and the method in~\cite{BaurBBG2011} is shown in \cref{fig:ThermalSurfacePlot} and \cref{fig:ThermFreqErrs}. In \Cref{fig:ThermalSurfacePlot}, we depict the $\hinf$ errors for the considered parameter configurations. The $\hinf$ errors of our ROMs are significantly lower than the errors of the pIRKA model -- both for the simpler setup \cref{romdef:thermala} and for the more involved setup \cref{romdef:thermalb}. Note that the errors of pIRKA are lower at the parameter samples $p^{(1)}$ and $p^{(2)}$ and increase further away from the samples. In contrast to that, SOBMOR provides low and roughly constant error surfaces. We use the method in~\cite{AliyevBMSV2017} to compute the $\hinf$ errors, which is designed for the computation of $\hinf$ norms of large-scale transfer functions, due to the vast amount of otherwise costly $\hinf$ norm computations.

\begin{figure}[tb]
  \centering
  \input{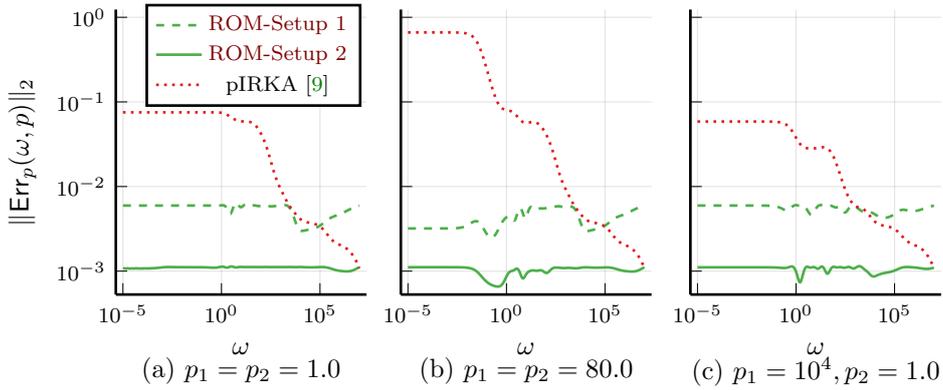}
  \caption{Frequency-wise errors as $p_1$ and $p_2$ are varied.}\label{fig:ThermFreqErrs}
\end{figure}

In \cref{fig:ThermFreqErrs}, we show the spectral norm of the error transfer function, $\ERR_p(\cdot, p)$ for different parameter configurations $p\in \Omega$. This highlights the fact that SOBMOR not only produces roughly constant $\hinf$ errors across the parameter range but also leads to rather flat transfer function errors at specific parameter configurations. 



\subsection{Timoshenko Beam}\label{sec:Timoshenko}
For the comparison of our method with the matrix-inter\-polation-framework, we use the benchmark model used in~\cite{GeussPL2013}, given by a 3D cantilever Timoshenko beam. The free parameter $p$ is the beam length, which is varied in $\Omega= [0.4,2.4]$.
The model is given by
\begin{align*}
  E(p) \dot x &= A(p) x + B u, \\
  y &= C x,
\end{align*}
where $E(p), A(p) \in \R^{240 \times 240}$ and $E(p)$ is a full-rank matrix for all $p \in [0.4, 2.4]$. In this example, and in contrast to the one considered in \Cref{sec:Thermal}, the parameter dependency of $E$ and $A$ on $p$ is nonlinear. In~\cite{GeussPL2013} six local ROMs are computed at six different lengths $p^{(i)}$, $i=1,\ldots,6$ uniformly distributed between $0.4$ and $2.4$. We use the \textsf{psssMOR} toolbox again to compute the ROMs\footnote{Our configuration of \textsf{psssMOR} for the interpolation-based PMOR of the Timoshenko beam is documented in detail at~\url{https://gist.github.com/Algopaul/1999b0e34b54f800f56cbdf1be1e45b4}}. For the local reduction, we use BT (which typically leads to low $\hinf$ errors) and a two-sided rational Krylov interpolation at zero as proposed in~\cite{GeussPL2013}. Both ROMs result in a reduced model of dimension $10$ as in~\cite{GeussPL2013}. The local ROMs are merged via piecewise linear interpolation.

For SOBMOR, we use an ansatz that also employs six ansatz functions to get a similar complexity as in the interpolated ROM in~\cite{GeussPL2013}.
\begin{romdef}\label{romdef:timo}
  We set the ROM dimensions to $r=10, n_u=1$, and $n_y=1$. We set all $\kappa_j = 6$ for $j \in \{J, R, Q, B, C, D\}$. Moreover, we define the ansatz functions $f_i$ for $i \in \{1, \dots 6\}$ by
\begin{align*}
  f_i &= \hat{f}(\cdot, 0.4+(i-2) \cdot 0.2, 0.4+i \cdot 0.2)
\end{align*}
and set
  $f_i^J \equiv f_i^R \equiv f_i^Q \equiv f_i^B \equiv f_i^C \equiv f_i^D \equiv f_i$ 
for $i \in \{1,\ldots,6\}$. Then our ROM is constructed as in~\eqref{eq:romform} and~\eqref{eq:ansatz}.
\end{romdef}
\begin{figure}[t]
  \centering
  \input{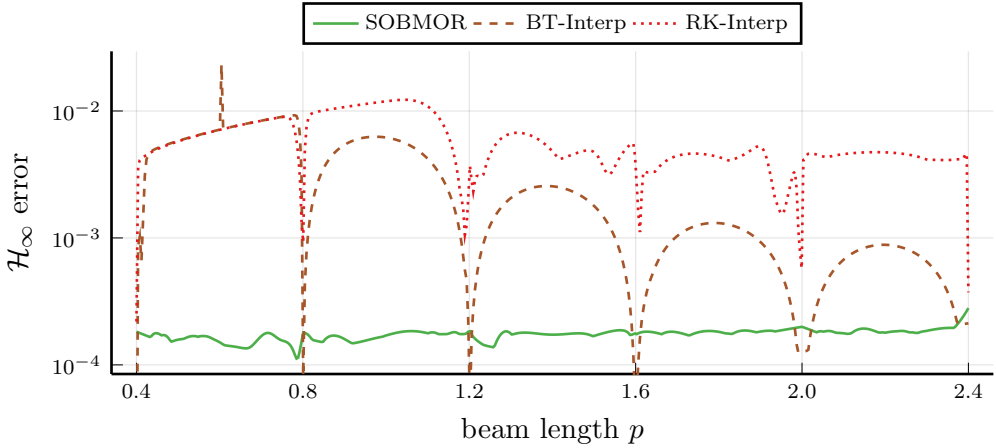}
  \caption{$\mathcal{H}_\infty $ error comparison for different PMOR methods.  }\label{fig:TimoshenkoHinfPlot}
\end{figure}
The results are shown in \cref{fig:TimoshenkoHinfPlot}.  The ROMs based on balanced truncation and rational Krylov interpolation are denoted by BT-Interp and RK-Interp, respectively. We observe that the interpolated ROMs are accurate at the interpolation points (in the case of BT-Interp even more accurate than the SOBMOR-ROMs). However, between the interpolation points, the error increases drastically, often by more than an order of magnitude. This is because in the interpolation framework, the ROM is not constructed to yield a small error in between the interpolation points. Therefore, the error can only be reduced in the matrix interpolation framework by increasing the number of interpolation points, which in turn increases the complexity of the ROM\@. On the other hand, using SOBMOR, we can set up a ROM ansatz and then minimize the error across a wide parameter range without changing the predefined ROM structure. In this way, we keep the complexity of the ROM low, while including a large number of parameter samples. This is the main benefit of using SOBMOR for parametric systems: a large number of parameter samples can be included during optimization but the ROM ansatz is chosen independently and can maintain a low complexity. Experiments analyzing the dependence of the approximation error on the size of the reduced model are conducted in the next subsection.

\subsection{Port-Hamiltonian MSD Chain}\label{sec:PHMSD}

Our final experiment is concerned with a port-Hamiltonian model of a mass-spring-damper chain from~\cite{Gugercin2012}.
We compare SOBMOR with IRKA-PH, which is adapted to parametric systems using the same technique as described \Cref{sec:prelim}. Moreover, we compare with the $\htwo \otimes \ltwo$ optimi\-zation-based approach~\cite{HundMMS2021}, which is not designed to preserve the pH structure. 

The model is described in detail in the package \textsf{PortHamiltonianBenchmarkSystems}\footnote{\url{https://algopaul.github.io/PortHamiltonianBenchmarkSystems.jl/SingleMSDChain/}}. To obtain a parametric model, we use the damping coefficient as free parameter in the interval $\Omega = [0.5, 1.5]$. Moreover, we only use a single input and output, since the implementation for the method in~\cite{HundMMS2021} only encompasses SISO systems.

The system equations are given by
\begin{align}
  \label{eq:msd_model}
  \msdM(p) :
  \begin{cases}
    \dot x = (\msdJ-\msdR(p))\msdQ x + \msdB u, \\
    y = \msdB^\T \msdQ x,
  \end{cases}
\end{align}
where $\msdJ, \msdQ, \msdR(p) \in \R^{100 \times 100}$ and $\msdB \in \R^{100 \times 1}$. The skew-symmetry of $\msdJ$, the positive definiteness and symmetry of $\msdQ$, and the positive semi-definiteness and symmetry of $\msdR(p)$ for all $p \in \R^+$ guarantee that $\msdM(p)$ is pH for all parameter values $p \in [0.5, 1.5]$.

To preserve the pH structure in our reduced order model, we do not use the ansatz in~\eqref{eq:romform} but instead define our ROM as follows.

\begin{romdef}\label{romdef:pH}
  We test different ROM orders $r \in \{1, \dots 10\}$ and set $n_u = 1$ and $n_y = 1$. Our ansatz is given by
\begin{align}
  \label{eq:msd_ansatz}
  \msdMred(c,\theta) :
  \begin{cases}
    \dot x = (\pJ(p, \theta)-\pR(p, \theta))\pQ(p, \theta) x + \pB(p, \theta) u, \\
    y = \pB{(p, \theta)}^\T \pQ(p, \theta) x,
  \end{cases}
\end{align}
where we set all $\kappa_j = 2$ for $j \in \{J, R, Q, B, C, D\}$. We define the scalar ansatz functions
\begin{align*}
  f_i &= \hat{f}(\cdot, 0.5+(i-2) \cdot 0.5, 0.5+i \cdot 0.5)
\end{align*}
and set $f_i^J \equiv f_i^R \equiv f_i^Q \equiv f_i^B \equiv f_i^C \equiv f_i^D \equiv f_i$ for $i \in \{1,2\}$. Then $\pJ, \pR, \pQ$ and $\pB$ are defined as in~\eqref{eq:ansatz}. The skew-symmetry of $\pJ^M(p, \theta)$ and the positive semi-definiteness of $\pR^M(p, \theta)$ and $\pQ^M(p, \theta)$ ensure that our ansatz leads to pH models for all considered $p\in \Omega$ and $\theta\in \mathbb{R}^{n_\theta}$.
\end{romdef}
We configure the method in~\cite{HundMMS2021} (denoted by $\mathcal{H}_2 \otimes \mathcal{L}_2$-Opt) using the most flexible ansatz provided in the available implementation\footnote{available at \url{https://zenodo.org/record/5710777}}, which allows for linear parameter dependencies in the system matrices.
\begin{romdef}
  The ROM structure used in $\mathcal{H}_2 \otimes \mathcal{L}_2$-Opt is given by
  \begin{align*}
    (E_{r, 1} + pE_{r, 2})\dot x &= (A_{r, 1} + pA_{r, 2})x + (B_{r, 1} + pB_{r, 2}) u, \\
    y &= (C_{r, 1} + pC_{r, 2})x,
  \end{align*}
  where we again test ROM orders $r \in \{1, \dots 10\}$.
\end{romdef}

For IRKA-PH, we use the sample points $\{0.5, 1.0, 1.5\}$ to obtain three projection matrices $V_{\text{IRKA}, i} \in \R^{100 \times 10}$, $i=1,\ldots,3$. The final projection matrix $V_{\text{IRKA}}$ is then obtained from a singular value decomposition of the horizontal concatenation of $V_{\text{IRKA}, i}$, $i=1,\ldots,3$. To obtain a ROM of order $r$, the first $r$ columns of $V_{\text{IRKA}}$ are used to compute the ROM\@.

\begin{figure}[t]
  \centering

\begin{tikzpicture}[/tikz/background rectangle/.style={fill={rgb,1:red,1.0;green,1.0;blue,1.0}, draw opacity={1.0}}, show background rectangle]
\begin{axis}[point meta max={nan}, point meta min={nan}, legend cell align={left}, legend columns={3}, title={}, title style={at={{(0.5,1)}}, anchor={south}, font={{\fontsize{14 pt}{18.2 pt}\selectfont}}, color={rgb,1:red,0.0;green,0.0;blue,0.0}, draw opacity={1.0}, rotate={0.0}}, legend style={color={rgb,1:red,0.0;green,0.0;blue,0.0}, draw opacity={1.0}, line width={1}, solid, fill={rgb,1:red,1.0;green,1.0;blue,1.0}, fill opacity={1.0}, text opacity={1.0}, font={{\fontsize{8 pt}{10.4 pt}\selectfont}}, text={rgb,1:red,0.0;green,0.0;blue,0.0}, cells={anchor={center}}, at={(0.5, 1.02)}, anchor={south}}, axis background/.style={fill={rgb,1:red,1.0;green,1.0;blue,1.0}, opacity={1.0}}, anchor={north west}, xshift={1.0mm}, yshift={-1.0mm}, width={0.8\textwidth}, height={0.35\textwidth}, scaled x ticks={false}, xlabel={reduced model order $r$}, x tick style={color={rgb,1:red,0.0;green,0.0;blue,0.0}, opacity={1.0}}, x tick label style={color={rgb,1:red,0.0;green,0.0;blue,0.0}, opacity={1.0}, rotate={0}}, xlabel style={}, xmajorgrids={true}, xmin={0.73}, xmax={10.27}, xticklabels={{$2$,$4$,$6$,$8$,$10$}}, xtick={{2.0,4.0,6.0,8.0,10.0}}, xtick align={inside}, xticklabel style={font={{\fontsize{8 pt}{10.4 pt}\selectfont}}, color={rgb,1:red,0.0;green,0.0;blue,0.0}, draw opacity={1.0}, rotate={0.0}}, x grid style={color={rgb,1:red,0.0;green,0.0;blue,0.0}, draw opacity={0.1}, line width={0.5}, solid}, axis x line*={left}, x axis line style={color={rgb,1:red,0.0;green,0.0;blue,0.0}, draw opacity={1.0}, line width={1}, solid}, scaled y ticks={false}, ylabel={$\mathcal{H}_\infty \otimes \mathcal{L}_\infty$ error}, y tick style={color={rgb,1:red,0.0;green,0.0;blue,0.0}, opacity={1.0}}, y tick label style={color={rgb,1:red,0.0;green,0.0;blue,0.0}, opacity={1.0}, rotate={0}}, ylabel style={}, ymode={log}, log basis y={10}, ymajorgrids={true}, ymin={0.0009435109541374004}, ymax={1.0225990690930593}, yticklabels={{$10^{-3}$,$10^{-2}$,$10^{-1}$,$10^0$}}, ytick={{0.001,0.01,0.1,1}}, ytick align={inside}, yticklabel style={font={{\fontsize{8 pt}{10.4 pt}\selectfont}}, color={rgb,1:red,0.0;green,0.0;blue,0.0}, draw opacity={1.0}, rotate={0.0}}, y grid style={color={rgb,1:red,0.0;green,0.0;blue,0.0}, draw opacity={0.1}, line width={0.5}, solid}, axis y line*={left}, y axis line style={color={rgb,1:red,0.0;green,0.0;blue,0.0}, draw opacity={1.0}, line width={1}, solid}, colorbar={false}]
    \addplot[color={rgb,1:red,0.302;green,0.6863;blue,0.2902}, name path={f7c8c7f6-7bcb-427e-98af-1b4ec95c9362}, draw opacity={1.0}, line width={1}, solid, mark={*}, mark size={3.0 pt}, mark repeat={1}, mark options={color={rgb,1:red,0.0;green,0.0;blue,0.0}, draw opacity={1.0}, fill={rgb,1:red,0.302;green,0.6863;blue,0.2902}, fill opacity={1.0}, line width={0.75}, rotate={0}, solid}]
        table[row sep={\\}]
        {
            \\
            1.0  0.14193762339405466  \\
            2.0  0.07566017348889567  \\
            3.0  0.036847361748898515  \\
            4.0  0.023490082976028382  \\
            5.0  0.006740613333747307  \\
            6.0  0.003220530578957909  \\
            7.0  0.0023112397743184467  \\
            8.0  0.002164614399696213  \\
            9.0  0.0011923746254673961  \\
            10.0  0.0011465080762474867  \\
        }
        ;
    \addlegendentry {SOBMOR$ $}
    \addplot[color={rgb,1:red,0.5961;green,0.3059;blue,0.6392}, name path={e136fb16-a2fd-4413-8553-a277f16ad7c7}, draw opacity={1.0}, line width={1}, solid, mark={star}, mark size={3.0 pt}, mark repeat={1}, mark options={color={rgb,1:red,0.0;green,0.0;blue,0.0}, draw opacity={1.0}, fill={rgb,1:red,0.5961;green,0.3059;blue,0.6392}, fill opacity={1.0}, line width={0.75}, rotate={0}, solid}]
        table[row sep={\\}]
        {
            \\
            1.0  0.25084166608134967  \\
            2.0  0.20055790086751188  \\
            3.0  0.042077626248238964  \\
            4.0  0.06122102250511263  \\
            5.0  0.051663164892689295  \\
            6.0  0.0233757251598945  \\
            7.0  0.020382935867183843  \\
            8.0  0.015455785366548014  \\
            9.0  0.016502350849391112  \\
            10.0  0.015637353191798365  \\
        }
        ;
    \addlegendentry {$\mathcal{H}_2 \otimes \mathcal{L}_2$-Opt \cite{HundMMS2021}}
    \addplot[color={rgb,1:red,1.0;green,0.498;blue,0.0}, name path={753e71df-d933-4dd4-9e55-e32b027558c2}, draw opacity={1.0}, line width={1}, solid, mark={triangle*}, mark size={3.0 pt}, mark repeat={1}, mark options={color={rgb,1:red,0.0;green,0.0;blue,0.0}, draw opacity={1.0}, fill={rgb,1:red,1.0;green,0.498;blue,0.0}, fill opacity={1.0}, line width={0.75}, rotate={180}, solid}]
        table[row sep={\\}]
        {
            \\
            1.0  0.34023463965876344  \\
            2.0  0.23180071483495787  \\
            3.0  0.29586080321514197  \\
            4.0  0.261333007267056  \\
            5.0  0.7592465731382829  \\
            6.0  0.315350741247336  \\
            7.0  0.2915285296243695  \\
            8.0  0.2082325100589644  \\
            9.0  0.2075999502746541  \\
            10.0  0.2742615676384892  \\
        }
        ;
    \addlegendentry {IRKAPH$ $}
\end{axis}
\end{tikzpicture} \\
  (a) $\hinflinf$ error comparison for different ROM orders

\begin{tikzpicture}[/tikz/background rectangle/.style={fill={rgb,1:red,1.0;green,1.0;blue,1.0}, draw opacity={1.0}}, show background rectangle]
\begin{axis}[point meta max={nan}, point meta min={nan}, legend cell align={left}, legend columns={2}, title={}, title style={at={{(0.5,1)}}, anchor={south}, font={{\fontsize{14 pt}{18.2 pt}\selectfont}}, color={rgb,1:red,0.0;green,0.0;blue,0.0}, draw opacity={1.0}, rotate={0.0}}, legend style={color={rgb,1:red,0.0;green,0.0;blue,0.0}, draw opacity={1.0}, line width={1}, solid, fill={rgb,1:red,1.0;green,1.0;blue,1.0}, fill opacity={1.0}, text opacity={1.0}, font={{\fontsize{8 pt}{10.4 pt}\selectfont}}, text={rgb,1:red,0.0;green,0.0;blue,0.0}, cells={anchor={center}}, at={(1.02, 1)}, anchor={north west}}, axis background/.style={fill={rgb,1:red,1.0;green,1.0;blue,1.0}, opacity={1.0}}, anchor={north west}, xshift={1.0mm}, yshift={-1.0mm}, width={0.8\textwidth}, height={0.35\textwidth}, scaled x ticks={false}, xlabel={reduced model order $r$}, x tick style={color={rgb,1:red,0.0;green,0.0;blue,0.0}, opacity={1.0}}, x tick label style={color={rgb,1:red,0.0;green,0.0;blue,0.0}, opacity={1.0}, rotate={0}}, xlabel style={}, xmajorgrids={true}, xmin={0.73}, xmax={10.27}, xticklabels={{$2$,$4$,$6$,$8$,$10$}}, xtick={{2.0,4.0,6.0,8.0,10.0}}, xtick align={inside}, xticklabel style={font={{\fontsize{8 pt}{10.4 pt}\selectfont}}, color={rgb,1:red,0.0;green,0.0;blue,0.0}, draw opacity={1.0}, rotate={0.0}}, x grid style={color={rgb,1:red,0.0;green,0.0;blue,0.0}, draw opacity={0.1}, line width={0.5}, solid}, axis x line*={left}, x axis line style={color={rgb,1:red,0.0;green,0.0;blue,0.0}, draw opacity={1.0}, line width={1}, solid}, scaled y ticks={false}, ylabel={$\mathcal{H}_2 \otimes \mathcal{L}_2$ error}, y tick style={color={rgb,1:red,0.0;green,0.0;blue,0.0}, opacity={1.0}}, y tick label style={color={rgb,1:red,0.0;green,0.0;blue,0.0}, opacity={1.0}, rotate={0}}, ylabel style={}, ymode={log}, log basis y={10}, ymajorgrids={true}, ymin={0.00064639779799002}, ymax={0.23073792842569033}, yticklabels={{$10^{-3}$,$10^{-2}$,$10^{-1}$}}, ytick={{0.001,0.01,0.1}}, ytick align={inside}, yticklabel style={font={{\fontsize{8 pt}{10.4 pt}\selectfont}}, color={rgb,1:red,0.0;green,0.0;blue,0.0}, draw opacity={1.0}, rotate={0.0}}, y grid style={color={rgb,1:red,0.0;green,0.0;blue,0.0}, draw opacity={0.1}, line width={0.5}, solid}, axis y line*={left}, y axis line style={color={rgb,1:red,0.0;green,0.0;blue,0.0}, draw opacity={1.0}, line width={1}, solid}, colorbar={false}]
    \addplot[color={rgb,1:red,0.302;green,0.6863;blue,0.2902}, name path={ef7289b3-6de0-4549-b26f-b8bfb7188f55}, draw opacity={1.0}, line width={1}, solid, mark={*}, mark size={3.0 pt}, mark repeat={1}, mark options={color={rgb,1:red,0.0;green,0.0;blue,0.0}, draw opacity={1.0}, fill={rgb,1:red,0.302;green,0.6863;blue,0.2902}, fill opacity={1.0}, line width={0.75}, rotate={0}, solid}]
        table[row sep={\\}]
        {
            \\
            1.0  0.1033086240445387  \\
            2.0  0.050804663465768875  \\
            3.0  0.02873137315239892  \\
            4.0  0.01803773103326831  \\
            5.0  0.0065255576056998275  \\
            6.0  0.0043213766987728825  \\
            7.0  0.0020297550624328027  \\
            8.0  0.0020212334527439364  \\
            9.0  0.0019219182527028553  \\
            10.0  0.0009593254240763519  \\
        }
        ;
    \addplot[color={rgb,1:red,0.5961;green,0.3059;blue,0.6392}, name path={5e4a9b03-d431-4c74-8cce-388d4aaded1c}, draw opacity={1.0}, line width={1}, solid, mark={star}, mark size={3.0 pt}, mark repeat={1}, mark options={color={rgb,1:red,0.0;green,0.0;blue,0.0}, draw opacity={1.0}, fill={rgb,1:red,0.5961;green,0.3059;blue,0.6392}, fill opacity={1.0}, line width={0.75}, rotate={0}, solid}]
        table[row sep={\\}]
        {
            \\
            1.0  0.07448545805818985  \\
            2.0  0.031461444984531274  \\
            3.0  0.014300097806513616  \\
            4.0  0.006256961006468841  \\
            5.0  0.0033929143996510014  \\
            6.0  0.0015641170766271242  \\
            7.0  0.0011058191723587597  \\
            8.0  0.0007633853227027855  \\
            9.0  0.0007800600545762375  \\
            10.0  0.0008066640456270851  \\
        }
        ;
    \addplot[color={rgb,1:red,1.0;green,0.498;blue,0.0}, name path={e6aed9ea-d337-4807-a640-ef6825513417}, draw opacity={1.0}, line width={1}, solid, mark={triangle*}, mark size={3.0 pt}, mark repeat={1}, mark options={color={rgb,1:red,0.0;green,0.0;blue,0.0}, draw opacity={1.0}, fill={rgb,1:red,1.0;green,0.498;blue,0.0}, fill opacity={1.0}, line width={0.75}, rotate={180}, solid}]
        table[row sep={\\}]
        {
            \\
            1.0  0.19537772657074548  \\
            2.0  0.17503908662964274  \\
            3.0  0.1728909583997595  \\
            4.0  0.13942590619624265  \\
            5.0  0.1472560042843293  \\
            6.0  0.09219082238718806  \\
            7.0  0.0874846260149802  \\
            8.0  0.06840099511303936  \\
            9.0  0.06839489948821997  \\
            10.0  0.05038292732602917  \\
        }
        ;
\end{axis}
\end{tikzpicture} \\
  (b) $\htwoltwo$ error comparison for different ROM orders
  \caption{Error comparison between SOBMOR and IRKA-PH on port-Hamiltonian mass-spring-damper chain example.}\label{fig:HinfxHinfPHMSD}
\end{figure}

In \cref{fig:HinfxHinfPHMSD}, we compare the accuracies of SOBMOR, IRKA-PH, and $\htwo \otimes \ltwo$-Opt. As expected, SOBMOR leads to higher accuracies in terms of the $\hinflinf$ error\footnote{We compute the $\hinf$ error for 200 parameter samples in $[0.5, 1.5]$ and plot the maximum $\hinf$ error of all samples as $\hinflinf$ error.}, while the method in~\cite{HundMMS2021} leads to lower $\htwoltwo$ errors. IRKA-PH leads to the worst accuracies in both norms. This can be explained by the fact that SOBMOR already finds more accurate ROMs in the nonparametric case (see~\cite{SchwerdtnerV2020}) and also seems to determine more accurate parametric models as demonstrated in our previous experiments. Note that in the nonparametric case, IRKA-PH has been vastly improved by utilizing different energy-representations~\cite{BreU21}.


The main drawback of our approach in comparison to the well-established PMOR methods is its longer (offline) runtime needed to compute the ROM\@. The runtime is mainly influenced by the number of optimization parameters, i.\,e.\ the order of the ROM and the number of ansatz functions. 
In \cref{fig:runtimePHMSD}, we show the runtimes of SOBMOR and $\htwoltwo$-Opt for the different reduced model orders. Both runtimes tend to increase, as the reduced model order $r$ increases, which is expected, since the dimension of the optimization parameter vector increases with $r$. We note that SOBMOR is often significantly faster than $\htwoltwo$-Opt but both methods can take up to several hours to complete. This long runtime (the well-established PMOR methods typically terminate within a few seconds) is the main drawback of the optimization-based approaches. However, in many applications of PMOR, a computationally expensive offline phase is tolerated to achieve fast and reliable predictions from a small and accurate parametric ROM\@~\cite{AllaireKLMW2014, MaininiW2015}.



\begin{figure}[htpb]
  \centering
  \input{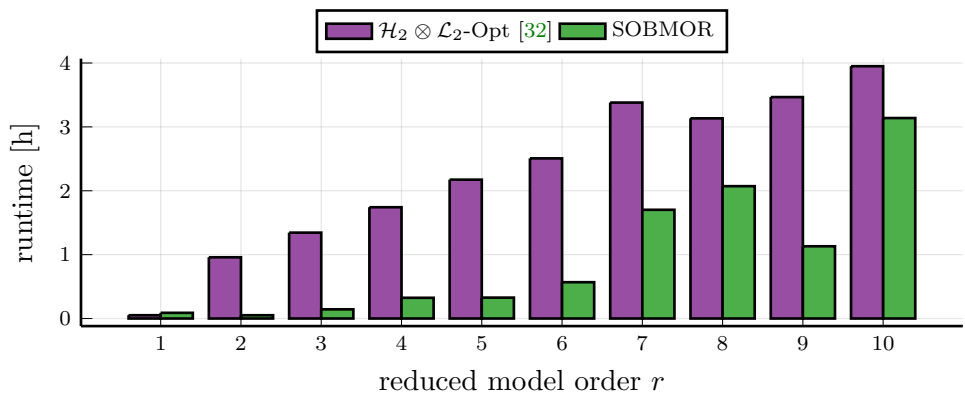}
\caption{Runtimes of SOBMOR for the port-Hamiltonian mass-spring-damper chain example for different reduced model orders. Experiments were performed on a Desktop PC with an \emph{Intel\textsuperscript{\textregistered} Core\textsuperscript{\texttrademark} i9--9900K} CPU at 3.60GHz and 32 GB of RAM\@.}\label{fig:runtimePHMSD}
\end{figure}

\section{Conclusion}\label{sec:conclusion}
We have presented a parametric model order reduction method by extending SOBMOR to parametric systems. For this, we have provided a new parameterization of stable parametric systems and an extension of an adaptive sampling strategy to the multi-dimensional case. Several numerical experiments demonstrate the high accuracy of our method in a comparison to state-of-the-art PMOR methods. 

We briefly present perspectives of future research. Model reduction of nonparametric descriptor systems (i.\,e.\ systems with a singular $E$-matrix) using SOBMOR was pursued in~\cite{MosSMV2022b, MosSMV22}. Such systems may have a transfer function containing a polynomial part. If this polynomial part is constant, we can apply our methodology with no further changes. In case of higher order polynomials, that may occur for higher index systems, the polynomial parts must be matched \emph{exactly} in the ROM to obtain a finite $\hinf$ error. This also applies to PMOR\@: If the polynomial part of the FOM transfer function has degree less than one for all $p \in \pdom$, we can straightforwardly extend the method presented in this article. Otherwise, the polynomial part must be matched exactly for all $p \in \pdom$, for which the degree is greater or equal to one, which requires a priori knowledge of the polynomial part of the FOM transfer function. 

Another research perspective concerns the choice of the ansatz functions in~\eqref{eq:ansatz}. Until now, we have only used univariate and linear hat functions. However, the ansatz functions may be designed specifically for a given use-case which allows for a more efficient ansatz with fewer parameters that need to be tuned. We currently investigate such a tailored parameterization for model reduction in real-time optimization-based retinal laser treatment~\cite{SchallerWKMBMW2022}.

\section*{Acknowledgments}
We gratefully acknowledge Volker Mehrmann and Benjamin Unger for their helpful comments on an earlier version of this manuscript.

\bibliographystyle{siamplain}
\bibliography{references}

\appendix
\section{Gradient Computation}\label{sec:gradcomp}

First, we restate a preliminary result that will be frequently used in the proof of Theorem~\ref{thm::gradients}. 
\begin{lemma}[{\cite[Lemma 3.2]{SchwerdtnerV2020}}]\label{lem::trace}
  Let $A \in \C^{m \times n}$ and let $e^{(j)}_i \in \C^j$ denote the $i$-th standard basis vector of $\C^{j}$. Then
  \begin{align}
    \trace\left(A\vtf_m\big(e_i^{(nm)}\big)\right)&=\big(e_i^{(nm)}\big)^\T \ftv\big(A^\T\big)\label{eq::tracevtf}.
\intertext{Further, setting $m=n$ and defining $n_1 := \frac{n(n+1)}{2}$ and $n_2 := \frac{n(n-1)}{2}$, we have}
    \trace\left(A\vtu\big(e_i^{(n_1)}\big)\right) &= \big(e_i^{(n_1)}\big)^\T \utv\big(A^\T\big), \label{eq::tracevtu}\\
    \trace\left(A\vtsu\big(e_i^{(n_2)}\big)\right) &= \big(e_i^{(n_2)}\big)^\T \sutv\big(A^\T\big)\label{eq::tracevtsu}.
  \end{align}
\end{lemma}

\begin{myproof}[Proof of Theorem~\ref{thm::gradients}]
	We start with the gradient $\gradB$ with respect to the input matrix parameterization. To this end, fix $i\in \{1,\,\dots,\,\kappa\cdot\dimr\cdot\dimu\}$ and let $e_i$  be the $i$-th standard basis vector of $\R^{\kappa\cdot\dimr\cdot\dimu}$. Consider the perturbation of the input matrix parameterization $\delta\theta_i(\epsilon) := \begin{bmatrix} \epsilon e_i^\T,\,0_{n_\theta - \kappa\cdot \dimr\cdot\dimu} \end{bmatrix}^\T~\in~\R^{n_\theta}$. Further set $\delta \pB(p^{(0)}) := \pB(p^{(0)},e_i)$ and compute
	\begin{align*}
	\ptf(s_0,p^{(0)};\theta_0+\delta\theta_i(\varepsilon)) = \ptf(s_0,&p^{(0)};\theta_0) \\&+ \varepsilon \pC(p^{(0)}, \theta_0){\left(sI-\pA(p^{(0)}, \theta_0)\right)}^{-1}\delta\pB(p^{(0)}),
	\end{align*}
	which clearly is differentiable in $\varepsilon$.
	Thus, by our nonzero and simplicity assumption on the $j$-th singular value of $\tf(s_0,p^{(0)})-\ptf(s_0,p^{(0)};\theta_0)$, we obtain the differentiability of $\varepsilon \mapsto \sigma_j(\tf(s_0,p^{(0)}) - \ptf(s_0,p^{(0)};\theta_0+\delta \theta_i(\varepsilon)))$.
	
	Abbreviating $\dyn= s_0I-\pA(p^{(0)},\theta_0)$, and using the invariance of the trace under cyclic permutations, we obtain
	\begin{align*}
	\tfrac{\text{d}}{\text{d}\varepsilon} \sigma_j(\tf(s_0,p^{(0)}) - \ptf(s_0,p^{(0)};\theta + \delta \theta_i(\varepsilon)) &=  -\Real(\uu^\H\pC(p^{(0)}, \theta_0)\dyn^{-1}\delta\pB(p^{(0)})\vv) \\
	&=-\Real(\trace(\vv\uu^\H\pC(p^{(0)}, \theta_0)\dyn^{-1}\delta\pB(p^{(0)}))).
	\end{align*}
	Further, $\delta \pB(p^{(0)}) := \pB(p^{(0)},e_i)= \sum_{i=1}^{\kappa}f_j^B(p^{(0)}) B_j(e_i)$ with $B_j(e_i) = \vtf_{\dimu}((e_i)_{B_j})$ if $i~\in~\{(j-1)\cdot \dimr \cdot \dimu,\ldots, j\cdot \dimr \cdot \dimu\}$ and zero otherwise, where $(e_i)_{B_i}$ denotes the part of $e_i\in \R^{\kappa\cdot\dimr\cdot\dimu}$ corresponding to the parameterization of $B_i$ according to the partitioning~\eqref{eq:partitioning}. Thus, for $j\in \N$ such that $i\in \{(j-1)\cdot \dimr \cdot \dimu,\ldots, j\cdot \dimr \cdot \dimu\}$ and invoking \cref{lem::trace} we have
	\begin{align*}
	\Real(\trace(\vv\uu^\H\pC(p^{(0)}\!, \!\theta_0)\dyn^{-1}\delta\pB(p^{(0)}))) &\!=\! \Real(\trace(\vv\uu^\H\pC(p^{(0)}\!,\! \theta_0)\dyn^{-1}f_j^B(p^{(0)})\vtf_{\dimu}((e_i)_{B_j}))\\
	&\!=\! \Real((e_i)^\top \ftv((\vv\uu^\H\pC(p^{(0)}, \theta_0)\dyn^{-1}f_j^B(p^{(0)}))^\top))
	\end{align*}
	which proves~\eqref{eq:gradB}. The formula considering the output matrix~\eqref{eq:gradC} and the feedthrough matrix~\eqref{eq:gradD} can be computed analogously. 
	To prove \eqref{eq:gradR}, consider now a perturbation with respect to the parameterization $\delta\theta_i(\epsilon) := \begin{bmatrix} 0_{n_\theta - \kappa\cdot\dimr(\dimr+1)},\epsilon e_i^\T,0_{\kappa\cdot\dimr(\dimr+1)/2} \end{bmatrix}^\T~\in~\R^{n_\theta}$ dissipation matrix, where $e_i \in \R^{\kappa\cdot\dimr(\dimr+1)/2}$. Then we have
	\begin{align*}
	&\ptf(p^{(0)},\theta_0 + \delta\theta_i(\varepsilon)) \\&\!= \!\pC(p^{(0)}, \theta_0){\left(sI\!-\!(\pJ(p^{(0)},\theta_0) \!-\! \pR(p^{(0)},\theta_0\! +\! \delta\theta_i(\varepsilon)))\pQ(p^{(0)},\theta_0)\right)}^{-1}\pB(p^{(0)},\theta_0) \\&+ \pD(p^{(0)},\theta_0).
	\end{align*}
	Let $j \in \N$ such that $i\in \{n_\theta - (\kappa-j-1)\cdot\dimr(\dimr+1)/2,\ldots,n_\theta - (\kappa-j)\cdot\dimr(\dimr+1)/2\}$, i.e., the parameter perturbation via $\varepsilon e_i$ corresponds to the block of $\theta_{R_j}$ in view of the partitioning \eqref{eq:partitioning}. Then, 
	\begin{align*}
	\pR(p^{(0)},\theta_0 + \delta\theta_i(\varepsilon)) = \sum_{k=1,k\neq j}^{\kappa} f_k(p^{(0)})R_k(\theta_0) + f_j^R(p^{(0)})R_j(\theta_0 + \delta \theta_i(\varepsilon)),
	\end{align*}
	where \begin{align}
	R_j(\theta_0 + \delta \theta_i(\varepsilon)) &= \vtu((\theta_0 + \delta \theta_i(\varepsilon))_{R_j}) \vtu((\theta_0 + \delta \theta_i(\varepsilon))_{R_j})^\top \nonumber\\
	&= \vtu(\theta_0)\vtu(\theta_0)^\top  + \varepsilon(\vtu((e_i)_{R_j})\vtu(\theta_0)^\top +\vtu(\theta_0) \vtu((e_i)_{R_j})^\top)\label{eq:Rperturbation}
	\\&\qquad \qquad \qquad + \varepsilon^2(\vtu((e_i)_{R_j})\vtu((e_i)_{R_j})^\top).\nonumber
	\end{align}
	Thus, setting $\delta R := (\vtu((e_i)_{R_j})\vtu(\theta_0)^\top +\vtu(\theta_0) \vtu((e_i)_{R_j})^\top)$ and abbreviating $\dyn = sI - \pA(p^{(0)},\theta_0) = sI - (\pJ(p^{(0)},\theta_0) - \pR(p^{(0)},\theta_0))\pQ(p^{(0)},\theta_0)$, we get
	\begin{align*}
	(sI-(\pJ(&p^{(0)},\theta_0) - \pR(p^{(0)},\theta_0 + \delta\theta_i(\varepsilon)))\pQ(p^{(0)},\theta_0)\\
	&= \dyn + \varepsilon f_j^R(p^{(0)})\left(\delta R + \varepsilon(\vtu((e_i)_{R_j})\vtu((e_i)_{R_j})^\top\right)\pQ(p^{(0)},\theta_0) \\
	&= \dyn\left(I + \dyn^{-1}\varepsilon f_j^R(p^{(0)})\left(\delta R + \varepsilon(\vtu((e_i)_{R_j})\vtu((e_i)_{R_j})^\top\right)\pQ(p^{(0)},\theta_0)\right)
	\end{align*}
	and hence, choosing $\varepsilon>0$ small enough and applying a Neumann argument to the right-hand side,
	\begin{align*}
	\left(sI-\right.&\left.(\pJ(p^{(0)},\theta_0) - \pR(p^{(0)},\theta_0 + \delta\theta_i(\varepsilon)))\pQ(p^{(0)},\theta_0)\right)^{-1} 
	\\& =\left(I + \dyn^{-1}\varepsilon f_j^R(p^{(0)})\left(\delta R + \varepsilon(\vtu((e_i)_{R_j})\vtu((e_i)_{R_j})^\top\right)\pQ(p^{(0)},\theta_0)\right)^{-1}\dyn^{-1}
	\\&= \sum_{l=0}^{\infty}\left(-\dyn^{-1}\varepsilon f_j^R(p^{(0)})\left(\delta R + \varepsilon(\vtu((e_i)_{R_j})\vtu((e_i)_{R_j})^\top\right)\pQ(p^{(0)},\theta_0)\right)^l \dyn^{-1}.
	\end{align*}
	Thus, 
	\begin{align*}
	\ptf(p^{(0)},\theta_0 + \delta\theta_i(\varepsilon)) = &\ptf(p^{(0)},\theta_0) - \pC(p^{(0)},\theta_0)\!\sum_{l=1}^{\infty}\!\!\left(\dyn^{-1}\!\varepsilon f_j^R(p^{(0)})\left(\delta R \!\right. \right.\\&+\left.\left.\! \varepsilon(\vtu((e_i)_{R_j})\!\vtu((e_i)_{R_j})^\top\right)\!\pQ(p^{(0)},\theta_0)\right)^l\!\! \dyn^{-1}\pB(p^{(0)},\theta_0)
	\end{align*}
	which, again due to the Neumann series argument, is differentiable for small $\varepsilon>0$, implying the differentiability of the map $\varepsilon \mapsto \sigma_j(\tf(s_0,p^{(0)}) - \ptf(s_0,p^{(0)};\theta_0+\delta \theta_i(\varepsilon)))$. Moreover,
	\begin{align*}
&	\tfrac{\text{d}}{\text{d}\varepsilon} \sigma_j(\tf(s_0,p^{(0)}) - \ptf(s_0,p^{(0)};\theta + \delta \theta_i(\varepsilon)) \Big|_{\varepsilon = 0}\\
	&\!=\! \Real(\uu^\H\pC(p^{(0)},\theta_0)\dyn^{-1} f_j^R(p^{(0)})\delta R\pQ(p^{(0)},\theta_0) \dyn^{-1}\pB(p^{(0)},\theta_0)\vv)
	\\&\!=\! \Real(\trace(\pQ(p^{(0)},\theta_0) \dyn^{-1}\!\pB(p^{(0)},\theta_0)\vv\uu^\H\pC(p^{(0)},\theta_0)\dyn^{-1} f_j^R(p^{(0)})\vtu((e_i)_{R_j})\vtu(\theta_0)^\top)) \\&+\!\Real(\trace(\pQ(p^{(0)}\!,\theta_0) \dyn^{-1}\pB(p^{(0)},\theta_0)\vv\uu^\H\pC(p^{(0)},\theta_0)\dyn^{-1} f_j^R(p^{(0)})\vtu(\theta_0) \vtu((e_i)_{R_j})^\top))
	\end{align*}
	Thus, using \cref{lem::trace}, we get
	\begin{align*}
	&\Real(\trace(\pQ(p^{(0)},\theta_0) \dyn^{-1}\pB(p^{(0)},\theta_0)\vv\uu^\H\pC(p^{(0)},\theta_0)\dyn^{-1} f_j^R(p^{(0)})\vtu((e_i)_{R_j})\vtu(\theta_0)^\top))\\
	&=\!\Real(\trace(\vtu(\theta_0)^\top\pQ(p^{(0)},\theta_0) \dyn^{-1}\pB(p^{(0)},\theta_0)\vv\uu^\H\pC(p^{(0)},\theta_0)\dyn^{-1} f_j^R(p^{(0)})\vtu((e_i)_{R_j}))\\
	&=\! \Real((e_i)_{R_i}^\top \utv ((Y_1f_j^R)^\top\vtu(\theta_0))
	\end{align*}
	with $Y_1 =\pQ(p^{(0)},\theta_0)\dyn^{-1}  \pB(p^{(0)},\theta_0) \vv \uu^\H \pC(p^{(0)},\theta_0) \dyn^{-1}$ and using that the trace does not change under transposition we obtain
	\begin{align*}
	&\Real(\trace(\pQ(p^{(0)},\theta_0) \dyn^{-1}\pB(p^{(0)},\theta_0)\vv\uu^\H\pC(p^{(0)},\theta_0)\dyn^{-1} f_j^R(p^{(0)})\vtu(\theta_0) \vtu((e_i)_{R_j})^\top))\\
	&= \!\Real( \trace(\left(\!\pQ(p^{(0)}\!,\theta_0) \dyn^{-1}\!\pB(p^{(0)}\!,\theta_0)\vv\uu^\H\pC(p^{(0)}\!,\theta_0)\dyn^{-1}\! f_j(p^{(0)})\!\vtu(\theta_0))\right)^\top \!\!\vtu((e_i)_{R_j})))\\
	&= \!\Real((e_i)_{R_i}^\top \utv(Y_1f_j^R\vtu(\theta_0))),
	\end{align*}
	which shows \eqref{eq:gradR}.
	The result for the skew symmetric part follows similarly \eqref{eq:gradJ}. 
	Last, we consider the gradient w.r.t.\ the parameterization of the self-adjoint matrix function $\pQ(p^{(0)},\theta)$. To this end, set $\delta\theta_i(\epsilon) := \begin{bmatrix} 0_{n_\theta-\kappa\dimr(\dimr+1)/2},\, \epsilon e_i^\T \end{bmatrix}^\T~\in~\R^{n_\theta}$, $e_i \in \R^{\kappa\dimr(\dimr+1)/2}$ where, analogously to \eqref{eq:Rperturbation}, we have
	\begin{align*}
	Q_j(\theta_0 + \delta \theta_i(\varepsilon))
	&= \vtu(\theta_0)\vtu(\theta_0)^\top  + \varepsilon(\vtu((e_i)_{Q_j})\vtu(\theta_0)^\top +\vtu(\theta_0) \vtu((e_i)_{Q_j})^\top)
	\\&\qquad \qquad \qquad + \varepsilon^2(\vtu((e_i)_{Q_j})\vtu((e_i)_{Q_j})^\top).
	\end{align*}
	Thus, setting $\delta Q = \vtu((e_i)_{Q_j})\vtu(\theta_0)^\top +\vtu(\theta_0) \vtu((e_i)_{Q_j})^\top$, we obtain
	\begin{align*}
	&(sI- (\pJ(p^{(0)},\theta_0)- \pR(p^{(0)},\theta_0))\pQ(p^{(0)},\theta_0+\delta\theta_i(\varepsilon))\\
	&= \!\dyn - (\pJ(p^{(0)},\theta_0) - \pR(p^{(0)},\theta_0))(f_j^Q(p^{(0)}) \varepsilon(\delta Q + \varepsilon (\vtu((e_i)_{Q_j})\vtu((e_i)_{Q_j})^\top))\\
	&= \!\dyn\!\left(\!I \!-\! \dyn^{-1}\!(\pJ(p^{(0)}\!,\theta_0)\!-\!\pR(p^{(0)}\!,\theta_0))(f_j^Q(p^{(0)}) \varepsilon(\delta Q \!+ \!\varepsilon (\vtu((e_i)_{Q_j})\vtu((e_i)_{Q_j})^\top))\right)
	\end{align*}
	and again, following a Neumann series argument, choosing $\varepsilon>0$ small enough, we have
	\begin{align*}
	&\left((sI-(\pJ(p^{(0)},\theta_0) - \pR(p^{(0)},\theta_0))\pQ(p^{(0)},\theta_0+\delta\theta_i(\varepsilon))\right)^{-1} 
	\\&\!= \!\sum_{l=0}^{\infty}\!\left(\!\dyn^{-1}\!(\pJ(p^{(0)}\!,\theta_0) \!- \! \pR(p^{(0)}\!,\theta_0))\varepsilon f_j^Q(p^{(0)})\!\left(\delta Q\! +\! \varepsilon(\vtu((e_i)_{Q_j}\!)\vtu((e_i)_{Q_j}\!)^\top \!)\right)\right)^l \!\! \dyn^{-1}\!.
	\end{align*}
	Thus,
	\begin{align*}
	\tfrac{\text{d}}{\text{d}\varepsilon}& \sigma_j(\tf(s_0,p^{(0)}) - \ptf(s_0,p^{(0)};\theta + \delta \theta_i(\varepsilon)) \Big|_{\varepsilon = 0}\\
	&= -\Real(\uu^\H\pC(p^{(0)},\theta_0)\dyn^{-1} (\pJ(p^{(0)},\theta) - \pR(p^{(0)},\theta_0))f_j^Q(p^{(0)})\delta Q \dyn^{-1}\pB(p^{(0)},\theta_0)\vv)\\
	&= -\Real( \trace(\dyn^{-1}\pB(p^{(0)},\theta_0)\vv\uu^\H\pC(p^{(0)},\theta_0)\dyn^{-1} (\pJ(p^{(0)},\theta) - \pR(p^{(0)},\theta_0))f_j^Q(p^{(0)}) \\&\qquad\qquad\qquad\qquad\qquad \qquad \left(\vtu((e_i)_{Q_j})\vtu(\theta_0)^\top+\vtu(\theta_0) \vtu((e_i)_{Q_j})^\top\right)))\\
	&= -\Real( (e_i)_{Q_j}^\top(\utv( (Y_2f_j^Q)^\top \vtu(\theta_0)+ Y_2f_j^Q\vtu(\theta_0))))
	\end{align*}
	with $Y_2 = \dyn^{-1}\pB(p^{(0)},\theta_0)\vv\uu^\H\pC(p^{(0)},\theta_0)\dyn^{-1} (\pJ(p^{(0)},\theta) - \pR(p^{(0)},\theta_0))$, which proves the formula~\eqref{eq:gradQ}. \\ $\phantom{1}\hfill \square$
\end{myproof}

\end{document}